\documentclass[1p,fleqn,authoryear,preprint]{elsarticle}

\usepackage{amsmath,amsfonts,amsthm,amssymb}
\usepackage[mathscr]{eucal}
\usepackage{eurosym}
\usepackage{graphicx}
\usepackage{ifpdf}

\newtheorem{theorem}{Theorem}
\newtheorem{proposition}{Proposition}
\newtheorem{corollary}{Corollary}

\newtheorem{lemma}{Lemma}
\newtheorem{question}{Question}

\def\FP{{the dollar trader }}
\def\AP{{the euro trader }}

\DeclareMathOperator*{\lcm}{lcm}

\newcommand{\calA}{\mathscr{A}}
\newcommand{\calR}{\mathscr{R}}

\newcommand{\calD}{\mathscr{D}}
\newcommand{\calI}{\mathscr{I}}
\newcommand{\bbG}{\mathbb{G}}
\newcommand{\bbD}{\mathbb{D}}
\newcommand{\bbR}{\mathbb{R}}

\newcommand{\bA}{\boldsymbol{A}}
\newcommand{\bD}{\boldsymbol{D}}
\newcommand{\bG}{\boldsymbol{G}}
\newcommand{\hbA}{\boldsymbol{\Hat{A}}}
\newcommand{\tbA}{\boldsymbol{\Tilde{A}}}
\newcommand{\ttbA}{\boldsymbol{\Tilde{\Tilde{A}}}}
\newcommand{\be}{\boldsymbol{e}}
\newcommand{\bv}{\boldsymbol{v}}

\ifpdf
\usepackage[bookmarks=false,%
pdfstartview=FitH,linkbordercolor={0.5 1 1},%
citebordercolor={0.5 1 0.5},unicode,%
pagebackref,hyperindex]{hyperref}%
\else
\fi

\newcommand{\zer}{\hphantom{-}0}

\makeatletter
\newenvironment{smallmatrix-mod}{\null\,\vcenter\bgroup
  \Let@\restore@math@cr\default@tag
  \baselineskip6\ex@ \lineskip2.5\ex@ \lineskiplimit\lineskip
  \ialign\bgroup\hfill$\m@th\scriptstyle##$\hfil&&\space\hfill
  $\m@th\scriptstyle##$\hfil\crcr
}{%
  \crcr\egroup\egroup\,%
}
\def\ps@pprintTitle{%
     \let\@oddhead\@empty
     \let\@evenhead\@empty
     \def\@oddfoot{\footnotesize\itshape~\hfill\today}%
     \let\@evenfoot\@oddfoot}
\makeatother

\begin{document}
\begin{frontmatter}
\title{Periodic Sequences of Arbitrage:
A Tale of Four Currencies\tnoteref{ak}}

\tnotetext[ak]{The authors are grateful to Andrew Caplin, Patrick
Minford, Stephen Ross and Dimitri Vayanos for useful comments or
suggestions. The usual disclaimer of responsibility applies. The
authors would also like to thank two anonymous referees of this
journal for their incisive and constructive suggestions as to the
revision of this paper.}

\author[uc]{Rod Cross}
\ead{rod.cross@strath.ac.uk}

\author[ittp]{Victor Kozyakin\corref{t1}}
\ead{kozyakin@iitp.ru}

\author[ucc]{Brian O'Callaghan}
\ead{briantoc@gmail.com}

\author[ucc]{Alexei Pokrovskii\corref{t1}\fnref{t2}}
\ead{A.Pokrovskii@ucc.ie}

\author[lce]{Alexey Pokrovskiy}
\ead{A.Pokrovskiy@lse.ac.uk}

\cortext[t1]{Authors were partially supported by the Federal Agency
for Science and Innovations of the Russian Federation (state
contract no. 02.740.11.5048)}

\fntext[t2]{It is with great sadness we report that Alexei
Pokrovskii died shortly after this paper had been completed.}

\address[uc]{Department of Economics,
University of Strathclyde,\\ Sir William Duncan Building, 130
Rottenrow, Glasgow, G4 OGE, Scotland}

\address[ittp]{Institute for Information Transmission Problems, Russian Academy of Sciences,\\
Bolshoj Karetny lane 19, Moscow
127994 GSP-4, Russia}

\address[ucc]{Department of Applied Mathematics
University College Cork, Ireland}

\address[lce]{London School of Economics and
Political Science,\\ Houghton Street, London WC2A 2AE, UK}

\begin{abstract}
This paper investigates arbitrage chains involving four currencies
and four foreign exchange trader-arbitrageurs. In contrast with the
three-currency case, we find that arbitrage operations when four
currencies are present may appear periodic in nature, and not
involve smooth convergence to a ``balanced'' ensemble of exchange
rates in which the law of one price holds.  The goal of this
article is to understand some interesting features of sequences of
arbitrage operations, features which might well be relevant in
other contexts in finance and economics.
\end{abstract}


\end{frontmatter}

\section{Introduction}\label{S-intro}

An arbitrage operation involves buying some good or asset for a
lower price than that for which it can be sold, taking advantage of
any imbalance in the quoted prices.  The ``law of one price'' is a
statement of a key implication of the absence of arbitrage
opportunities. In turn arbitrage is often the process invoked to
explain why goods or assets  that are in some sense ``identical''
should have a common price.

A study of commodity prices since 1273 concluded that ``\ldots
despite the steady decline in transportation costs over the past
700 years, the repeated intrusion of wars and disease, and the
changing fashions of commercial policy, the volatility and
persistence of deviations in the law of one price have remained
quite stable'' \cite[p. 18]{rogoff1996}. The present paper
investigates a relatively neglected complication regarding
arbitrage operations, namely the order in which information about
arbitrage opportunities is presented, illustrating this in relation
to arbitrage chains involving four currencies. The key finding is
that arbitrage operations can be periodic in nature, rather than
involving a smooth convergence to a law of one price.

The early literature on the law of one price is coeval with the
purchasing power parity explanation of foreign exchange rates. The
terminology was coined in \cite{Cassel1916}, involving arbitrage
between relatively homogeneous goods priced in different currencies
\cite{rogoff1996,rogoff2001}. Empirical tests suggest that
arbitrage operations in goods do not exert a strong influence on
exchange rates until the price index deviations involved exceed
about 25\% \cite{Engel1999,obstfeld2001}. Innovations that were
expected to reduce price dispersion, such as the European Single
Market legislation coming into effect in 1992, and the Economic and
Monetary Union project beginning in 1999, have had little effect on
price level disparities \cite{wolf2003}. The degree of price level
dispersion between US cities has displayed no marked trend over
time \cite{Rogers2001}. A study of the prices charged for identical
products in IKEA stores in twenty-five countries revealed typical
price divergences of 20--50\%, differences that could not be
attributed to just country or location-specific factors
\cite{Haskel2001}. Among the most cited reasons for deviations from
the law of one price are transaction costs, taxes, transport costs,
trade barriers, the costs of searching for price differences,
nominal price rigidities, customer market pricing, nominal exchange
rate rigidities and differences in market power \cite{taylor2002}.

In relation to assets, an early application of the law of one price
was to the interest rate parity theory of the forward exchange
rate, whereby the ratio of the forward to spot exchange rate
between two currencies is equal to the ratio of the interest rates
in the two currencies over the forward period in question
\cite[p.~130]{keynes1923}. An arbitrage opportunity in relation to
assets can be defined as ``an investment strategy that guarantees a
positive payoff in some contingency with no possibility of a
negative payoff and with no net  investment''
\cite[online]{dybvig2008}. The absence of such arbitrage
opportunities has been seen as the unifying concept underlying
mainstream theories in finance, no-arbitrage principles being
applied in the Modigliani--Miller theorem of corporate capital
structure, in the Black--Scholes model of option pricing and in the
arbitrage pricing model of asset prices \cite{ross1978}. Actual
arbitrage operations in relation to assets often involve net
investment and risk and/or uncertainty, in addition to the
complications arising in relation to arbitrage in goods. Notable
deviations from the law of one price in financial markets have been
documented in relation to comparable circumstances applying to
closed-end country funds, American Depository Receipts, twin
shares, dual share classes and corporate spin-offs
\cite{lamont2003}. Among the limits to arbitrage in financial
markets are those arising from transactions costs
\cite{deardorff1979}, and those involving the capital requirements
of conducting arbitrage operations \cite{shleifer1997}. A
spectacular illustration of the capital limits to arbitrage was
provided by the demise of the Long-Term Capital Management (LTCM)
hedge funds. The arbitrage discrepancies being exploited in LTCM's
``convergence trades'' widened in 1998. LTCM attempted
unsuccessfully to raise new capital to finance its arbitrage
positions. To avoid a major financial collapse the  New York
Federal Reserve Board organised a bail-out by creditors
\cite{lowenstein2000}.

In what follows we focus on the limits to arbitrage arising from
the order in which information is disseminated to arbitrage
traders. The illustration used is for a foreign exchange (FX)
market with four FX traders and four currencies, see
Sections~\ref{S-3currencies} and \ref{S-4currencies}. An Arbiter,
the metaphorical equivalent of an unpaid auctioneer in a Walrasian
system, knows all the actual exchange rates. The individual FX
traders, however, initially know only the exchange rates involving
their own, domestic currencies. Justification for the assumptions
used in our model is provided in Section~\ref{S-structFX}. So the
US FX trader knows the exchange rates for the dollar against the
euro, sterling and yen, but not the cross exchange rates for the
non-dollar currencies. There are no transactions costs, no net
capital requirements and no risks involved in the arbitrage
operations. Instead we focus on the information dissemination
problem, and show that the order in which information about cross
exchange rate discrepancies, and hence arbitrage opportunities, is
presented makes an important difference to the sequences of
arbitrage operations conducted.

A general discussion of arbitrage dynamics is given in
Section~\ref{S-arbitrages}. An unexpected feature of the processes
considered in this paper is that, \emph{rather than there being a
smooth convergence to an ensemble of exchange rates with no
arbitrage opportunities, the arbitrage operations may display
periodicity and no necessary convergence on a cross exchange rate
law of one price}. See Proposition~\ref{32} in Section~\ref{mrSS}
for a rigorous explanation. A further unexpected feature is that,
\emph{starting at an ensemble of exchange rates which is not
balanced, and using special periodic sequences of arbitrages, the
Arbiter can achieve \textrm{any} balanced (satisfying the law of
one price) exchange rate ensemble}. See, in particular,
Theorem~\ref{arbH} in Section~\ref{mrSS} and Theorem~\ref{irratBC}
in Section~\ref{S-gencase}. These counter-intuitive results are
new, as far as we are aware. In line with the renowned
``impossibility theorem'' of \cite{arrow1951} these results suggest
an ``arbitrage impossibility theorem''. Proofs are relegated to
Section~\ref{S-proofs}.

The mathematical approach taken in this paper to the analysis of
arbitrage operation chains may be understood as a typical example
of the asynchronous interactions that are important in systems
theory and in control theory, see the monographs
\citep{BertTsi:89,AKKK:92:e,KasBh:2000} and the surveys
\cite{Koz:BCRI03-13,Koz:ICDEA04}. The arbitrage chains are
particularly relevant to desynchronised systems theory, see
\cite{AKKK:92:e}. Presence of an asynchronous interaction often
leads to a dramatic complication of the related mathematical
problems.  \cite{Koz:AiT90:6:e,Koz:AiT03:9:e} proved  that many
asynchronous problems cannot be solved algorithmically, and also
\cite{BT:IPL97,BT:SCL00,BT:Autom00} and \cite{TB:MCSS97}
demonstrated that, even in the cases when the problem is
algorithmically solvable, it is typically as hard to solve
numerically as the famous ``Travelling salesman problem,'' see
\cite{NP2} (that is, in the mathematical language, the problem is
NP-hard which is an abbreviation for ``Non-deterministic
Polynomial-time hard'' which means in the theory of algorithms that
a problem is very hard, if possible, to solve,  see \cite{NP}). In
this context the fact that the principal questions that arise in
analysis of arbitrage operation chains admit straightforward
combinatorial analysis came to the authors as a pleasant surprise.
Our construction uses a geometrical approach to visualisation of
arbitrage chains presented in
Sections~\ref{simpleASS}--\ref{S-proofs}, which may be useful in
relation to other problems in mathematical economics.

The periodicity results in this paper have implications for several
strands of literature. One is that dealing with the disequilibrium
foundations of equilibrium economics. The stability analysis of
Fisher poses the question: ``can one expect to prove that an
economy with rational agents conscious of disequilibrium and taking
advantage of arbitrage opportunities is driven (asymptotically) to
any equilibrium, Walrasian or constrained?''
\cite[pp.~86--87]{Fisher}. Fisher uses the assumption of ``no
favorable surprise'' as a means of demonstrating that a cessation
of exogenous shocks can lead to convergence to equilibrium. The
results in this paper suggest that there can be endogenous reasons,
arising from the cyclical response of arbitrage sequences to an
exogenous shock that gives rise to an arbitrage opportunity, why
convergence to equilibrium may not take place.

Another strand of literature to which our results relate is that on
market segmentation and arbitrage networks. Goods and assets are
not traded on a single exchange. Instead there are various trading
posts, such as commodity and stock exchanges. Other trades,
including a sizeable proportion of foreign exchange deals, are
conducted ``over the counter'' in direct transactions that bypass
formal exchanges. ``As a result, various clienteles trade on
different exchanges, and very few retail clients trade on more than
one exchange, let alone on all of them simultaneously''
\cite[p.~3]{RZ}. A key aspect of segmentation in the foreign
exchange ``market'' is that dealing rooms tend to specialise in
domestic currency trades. This provides a rationale for the
specification in this paper that foreign exchange dealers initially
are aware of only the exchange rates involving their domestic
currencies. We restrict our analysis to the case of 4 currencies,
with 6 principal exchange rates. The Financial Times gives daily
quotes for 52 currencies. The 1,326 principal exchange rates
involved suggest richer potential opportunities for arbitrage than
in the four currency case studied in the present paper. Bank for
International Settlements (BIS) data indicate that, in 2010,
transactions in these four currencies counted for 155.9\%{} of
global FX market turnover, the currency components being US dollars
(84.9\%), euros (39.1\%), Japanese yen (19.0\%) and pound sterling
(12.9\%). Because two currencies are involved in each transaction,
the \%{} shares sum to 200\%{} \cite[Table~B.4]{BIS}.

\section{Micro Structure of the FX market}\label{S-structFX}

In a centralised market trade takes place at prices that are public
information and traders face the same potential trading
opportunities. In contrast the FX market is decentralised, with the
end-user bank customers, banks, brokers and central banks involved
facing several possible methods of executing transactions, and
possibly different exchange rate quotes, some of which constitute
private information. BIS data for FX spot exchange rate
transactions in 2010 \cite[Table~E.24]{BIS} indicate the following
breakdown in execution methods as a \%{} of total global turnover:
inter-dealer direct (14.9\%), customer direct (21.6\%), voice
broker (8.6\%), electronic broking system (26.0\%), single-bank
electronic proprietary trading platforms (14.3\%) and multi-bank
dealing systems (14.5\%). Until the late 1980s FX transactions were
conducted largely by telephone, with FX dealers phoning
counterparties to get bid (buy) and offer (sell) quotes for
specific transaction amounts, there also being indirect dealing via
voice brokers who would search for matching interests between
clients, see \cite{GH}. The last two decades have seen a growth in
electronic methods of execution, a distinction being between
electronic broking systems such as Reuters Matching and the
Electronic Broking System Spot Dealing System (EBS), and single or
multi-bank proprietary dealing platforms.

A burgeoning literature investigates how this fragmented trading
structure impacts on price determination in FX markets, see
\cite{Lyons} and \cite{Evans} for surveys. The key contrast is
between the decentralised transactions conducted by FX dealers who
quote bid and offer prices that are not public information, and the
one-way bid or offer limit orders to buy and sell currencies at a
specific price that are accumulated by FX brokers in the
quasi-centralised segment of the market. This means that market
information is fragmented, FX dealers having private information
about the transactions forthcoming at their own quoted bid and
offer prices, and having access to the public information regarding
the order flows accumulated by the FX brokers. For analysis,
discussion and evidence regarding how order flows impact on
intra-day exchange rates see \cite{EL02,ST,EL08}.

The analysis in the present paper  assumes that FX dealers
initially know only the exchange rates for their own domestic
currencies, the order in which they discover imbalances in the
exchange rate ensemble, in the form of cross exchange rate
discrepancies, playing a key role in the arbitrage sequences
conducted. The fragmented nature of information in FX markets
suggests that this strong assumption has a whiff of reality in that
different FX dealers are likely to have disjoint information sets
and can conduct trades at different prices. Individual FX dealers
conducting bi-lateral trades with end-users receive private
information in the form of the orders forthcoming at their quoted
bid and ask prices, and this can give rise to profitable arbitrage
opportunities. For example, an FX dealer specialising in US dollars
might simultaneously receive large buy orders for euros and large
sell orders for Japanese yen, and suspect that euros are
under-priced relative to Japanese yen. After checking out the euro
-- Japanese yen exchange rates quoted in the inter-dealer market,
or in the brokered section of the market where information is
public, the dealer might discover that this is indeed the case, and
exploit this arbitrage opportunity regarding which other FX traders
are initially unaware.

The BIS data on the geographical distribution of FX market turnover
is informative in relation to the assumption in the present paper
that FX dealers initially are aware of only the exchange rates
involving their own domestic currency. Banks located in the UK
account for 37\%{} of global FX turnover, followed by the US
(18\%), Japan (6\%), Singapore (5\%), Switzerland (5\%), Hong Kong
(5\%) and Australia (4\%) --- see \cite[Graph~B.7]{BIS}. Although
cross-border transactions account for nearly two-thirds of FX
market turnover, this still leaves 35\%{} of the turnover being
local in nature \cite[Table~3.2]{BIS}, suggesting that the tendency
of FX dealers initially to focus on the exchange rates involving
their own domestic currencies assumed in the present paper is
evident in a significant section of the FX market.

Traders could be better informed about exchange rate developments
involving their own domestic currencies for a variety of reasons.
This could be simply because their core end-users have the domestic
currency as a unit of account, and means of payment, so the
domestically-based FX dealers have a ``home bias'' when it comes to
the exchange rates that they consider first. The psychology
literature indicates that here are quite tight limits to the pieces
of information that the working memory can take into account when
decisions are made \cite{Bad04}, suggesting that there could well
be advantages to FX traders if they focus, at least initially, on a
limited number of exchange rates. Alternatively the ``home bias''
could be due to the existence of different time zones. So, for
example, Japanese FX traders may be more able to react to new
information relevant to the Japanese yen during the Asian trading
hours in which North American and European markets are closed. The
evidence is that most FX trades initiated in Japan and Australia
occur during Asian trading hours; most trades initiated in the US
and Canada occur during North American hours; while UK-initiated
trades tend to be bunched in the overlapping Asia -- Europe and
Europe -- North America time zones \cite[Table~2]{Souza}.  A
further reason for ``home bias'' is that the localised or
institutionalised links that FX traders have with domestic clients
gives them order flow information about the likely course of the
exchange rates involving the domestic currency before the price
impact of this information becomes publicly available, via the
effects on inter-dealer trades, to FX traders operating in foreign
locations.

In \cite{CM02} the authors pose the question ``does Tokyo know more
about the yen?''. Prior to December 22, 1994 the Japanese FX market
closed for lunch from 12--00 to 13--30 hours, Tokyo time. On the
basis that local order flow conveys informational advantages, the
authors postulate that the trades of informed Tokyo traders would
be bunched before the lunch-time FX market closure, an effect that
would disappear once the lunch-time closure was abolished. They
found a significant tendency of foreign quotes on the Japanese yen
-- US dollar market to lag behind the Tokyo quotes in this
pre-lunch period, suggesting either that Tokyo-based traders were
better informed about the Japanese yen -- US dollar exchange rate
than FX traders based in foreign locations, or that foreign-based
traders believed this to be the case. Further evidence for a ``home
bias'' in the FX market was found in a study of the Canadian dollar
-- US dollar and Australian dollar -- US dollar markets
\cite{Souza}. The author calculates the impulse response functions
of the exchange rates to trades, measured by the order flows,
initiated in different locations. Trades initiated in Canada had a
larger long-run impact on the Canadian dollar -- US dollar exchange
rate than those initiated the US during North American trading
hours, and than Australian and Japanese trades initiated during
Asian trading hours. UK-initiated trades had a slightly larger
long-run effect during European trading hours, but this effect was
much larger before the start of North American trading hours.
Somewhat similarly, trades initiated in Australia had a larger
long-run impact on Australian dollar -- US dollar exchange rate
than trades initiated in the US and elsewhere. The conclusion is
that ``dealers operating both at the same time and in the same
geographic region as fundamentally driven customers have a natural
informational advantage'' \cite[pp.~23--24]{Souza}.

A major challenge to theories based on the idea that macroeconomic
``fundamentals'' drive exchange rates was presented by the
Meese-Rogoff results that such models did not forecast any better
than the ``naive'' postulate that the exchange rate rate would
remain unchanged \cite{MR83}. Engel and West \cite{EW05} showed
that exchange rates would display something close to the random
walk implied by the naive forecast if the fundamentals followed an
$\pm(1)$ process and the factor for discounting future fundamentals
was close to one. The microstructure literature has shown a way out
of this impasse, showing that micro-based information regarding
order flows, information which is not necessarily publicly
available, can explain a significant component of exchange rate
variation. So, for example, Evans and Lyons \cite{EL05} show that
end-user order flow data can explain around 16\%{} of the variance
in the monthly spot rate between the US dollar and the euro,
outperforming both standard macro fundamentals models and the
random walk specification. The microstructure literature has also
focussed attention onto high frequency data sets. Osler
\cite{Osl05}, for example, analyses minute-by-minute quotes for the
US dollar spot exchange rates with the Deutschmark, Japanese yen
and pound sterling, discovering significant effects from stop-loss
order flows, where the stop-loss order is one that instructs FX
dealers to buy (sell) a certain amount of a currency at the
``market'' rate once the exchange rate has risen (fallen) to a
pre-specified level.

A full survey of the theoretical and empirical literature on FX
exchange rate determination has been beyond the scope of the
present paper (see \cite{Evans} for such a survey). What we would
argue is that the foregoing selective review of the literature
provides some justification for the assumptions used in the
analysis of the arbitrage sequences that follows. The assumption
that FX dealers initially know only the exchange rates for their
domestic currency finds some support in the ``home bias'' evidence
cited above. The evidence on the fragmented nature of the FX market
lends support to the assumption that FX traders can have privileged
access to initially private information, stemming from order flows
from end-user clients, that would allow them to identify imbalances
in cross exchange rates, and hence identify arbitrage opportunities
before FX traders based in other locations can identify such
opportunities. There is also evidence that there are arbitrage
opportunities to be exploited. In \cite{MTY} the authors use
binding quote and transactions data from the electronic broking
system, EBS, for the US dollar, euros, Japanese yen, the pound
sterling and the Swiss franc. Triangular arbitrage opportunities
are identified within two-minute time horizons, and can be
exploited by three trades on the EBS trading screen. Each
identified arbitrage opportunity involved the US dollar and the
euro, the third currency being the Japanese yen, pound sterling or
Swiss franc. The estimated mean arbitrage profits, net of bid-offer
spreads and 0.2 basis point trade fees, ranged from 2.8 to 3.0
basis points \cite[p.~4]{MTY}. So here is evidence of, albeit
small, profits to be had from arbitrage operations on a
quasi-centralised, electronic broking trading platform. Once the
decentralised sections of the FX market are considered, the
existence of initially private information is likely extend the
range and size of profitable arbitrage arbitrage opportunities
available.

\section{The Three Currency Case}\label{S-3currencies}

Consider a foreign exchange (FX) market that involves only three
currencies: Dollars (\$), Euros (\euro) and Sterling (\pounds).
This FX market involves three pair-wise exchange operations:
\[
Dollar  \leftrightarrows Euro, \quad Dollar  \rightleftarrows Sterling,
\quad Euro \rightleftarrows Sterling.
\]
The currencies are measured in natural currency units, and the
corresponding (strictly positive) exchange rates,
$r_{\textrm{\$\euro}}$,  $r_{\textrm{\$}\textrm{\pounds}}$,
$r_{\textrm{\euro\pounds}}$, are well defined. For instance, one
dollar can be exchanged for $r_{\textrm{\$\euro}}$ euros.
 The rates related to the inverted arrows
are reciprocal:
\begin{equation}\label{rec3}
r_{\textrm{\euro}\textrm{\$}}=
\frac{1}{r_{\textrm{\$}\textrm{\euro}}},\quad
r_{\textrm{\pounds}\textrm{\$}}=
\frac{1}{r_{\textrm{\$}\textrm{\pounds}}},\quad
r_{\textrm{\pounds}\textrm{\euro}}=
\frac{1}{r_{\textrm{\euro}\textrm{\pounds}}}\ .
\end{equation}
We treat the triplet
\begin{equation}\label{prr3}
\left(
 r_{\textrm{\$}\textrm{\euro}}, \
 r_{\textrm{\$}\textrm{\pounds}}, \
 r_{\textrm{\euro}\textrm{\pounds}}
\right)
\end{equation}
as the ensemble of \emph{principal exchange rates.}

We suppose that, prior to a reference time moment $0$, each FX
trader knows only the exchange rates involving his domestic
currency. So the dollar trader does not know the value of
$r_{\textrm{\euro}\textrm{\pounds}}$, the euro trader is unaware of
$r_{\textrm{\$}\textrm{\pounds}}$, and the sterling trader is
unaware of $r_{\textrm{\$}\textrm{\euro}}$. We are interested in
the case where the initial rates are unbalanced in the following
sense. By assumption, the dollar trader can exchange one dollar for
$r_{\textrm{\$}\textrm{\euro}}$ euros. Let us suppose that
unbeknownst to him the exchange rate between sterling and euro is
such that the the dollar trader could make a profit by first
exchanging a dollar for $r_{\textrm{\$}\textrm{\pounds}}$ units of
sterling and then exchanging these for euros. The inequality which
guarantees that dollar trader can take advantage of this arbitrage
opportunity is that the
product
$r_{\textrm{\$}\textrm{\pounds}}r_{\textrm{\pounds}\textrm{\euro}}$
is greater than $r_{\textrm{\$}\textrm{\euro}}$:
\begin{equation}\label{unb1}
r_{\textrm{\$}\textrm{\pounds}}\cdot
r_{\textrm{\pounds}\textrm{\euro}}>
r_{\textrm{\$}\textrm{\euro}}.
\end{equation}

Let us consider the situation where the inequality \eqref{unb1}
holds, and, after the reference time moment $0$, one of the three
traders becomes aware of the third exchange rate. The evolution of
this FX market depends on \emph{which trader is the first to
discover the information concerning the third exchange rate}. The
following three cases are relevant.

\subsection{Case 1.}
The dollar trader becomes aware of the value of the rate
$r_{\textrm{\euro}\textrm{\pounds}}$. Therefore, the dollar trader
contacts the euro trader and makes a request to increase the rate
$r_{\textrm{\$}\textrm{\euro}}$ to the new fairer value
\[
r^{new}_{\textrm{\$}\textrm{\euro}}=
r_{\textrm{\$}\textrm{\pounds}}\cdot
r_{\textrm{\pounds}\textrm{\euro}} =
\frac{r_{\textrm{\$}\textrm{\pounds}}}{r_{\textrm{\euro}\textrm{\pounds}}}.
\]
The reciprocal exchange rate  $r_{\textrm{\euro}\textrm{\$}}$ is
also to be adjusted to the new level:
\[
r^{new}_{\textrm{\euro}\textrm{\$}}=
\frac{1}{r^{new}_{\textrm{\$}\textrm{\euro}}}.
\]
The result is that the principal exchange rates become balanced
at the levels:
\[
r^{new}_{\textrm{\$}\textrm{\euro}}=
\frac{r_{\textrm{\$}\textrm{\pounds}}}{ r_{\textrm{\euro}\textrm{\pounds}}}, \quad r_{\textrm{\$}\textrm{\pounds}}, \quad
r_{\textrm{\euro}\textrm{\pounds}}.
\]

\subsection{Case 2.}
The euro trader is the first to discover the third exchange rate
$r_{\textrm{\$}\textrm{\pounds}}$. By \eqref{rec3}, inequality
\eqref{unb1} may be rewritten as
\[
\frac{r_{\textrm{\$}\textrm{\pounds}}}{r_{\textrm{\euro}\textrm{\pounds}}}<\frac{1}{r_{\textrm{\euro}\textrm{\$}}},
\]
which is, in turn,  equivalent to
$r_{\textrm{\euro}\textrm{\$}}\cdot
r_{\textrm{\$}\textrm{\pounds}}>
r_{\textrm{\euro}\textrm{\pounds}}$. In this case the euro trader
could do better by first exchanging euros for dollars,  and then
by exchanging the dollars for sterling. Therefore, the euro trader
requests adjustment of the rate
$r_{\textrm{\euro}\textrm{\pounds}}$ to the value
\[
r^{new}_{\textrm{\euro}\textrm{\pounds}}=
r_{\textrm{\euro}\textrm{\$}}\cdot r_{\textrm{\$}\textrm{\pounds}}= \frac{r_{\textrm{\$}\textrm{\pounds}}}{r_{\textrm{\$}\textrm{\euro}}}.
\]
In terms of the principal exchange rates  the outcome is that the
FX market adjusts to the following balanced rates:
\[
r_{\textrm{\$}\textrm{\euro}}, \quad
r_{\textrm{\$}\textrm{\pounds}}, \quad
r^{new}_{\textrm{\euro}\textrm{\pounds}}=
\frac{r_{\textrm{\$}\textrm{\pounds}}}{r_{\textrm{\$}\textrm{\euro}}}.
\]

\subsection{Case 3.}
The sterling trader is the first to discover the third exchange
rate $r_{\textrm{\$}\textrm{\euro}}$. The inequality \eqref{unb1}
may be rewritten as $ r_{\textrm{\pounds}\textrm{\euro}}\cdot
r_{\textrm{\euro}\textrm{\$}}>r_{\textrm{\pounds}\textrm{\$}}. $
Thus, the sterling trader   requests adjustment of the rate
$r_{\textrm{\pounds}\textrm{\$}}$ to
$r^{new}_{\textrm{\pounds}\textrm{\$}}=
r_{\textrm{\pounds}\textrm{\euro}}\cdot
r_{\textrm{\euro}\textrm{\$}}$.  In this case the principal
exchange rates become balanced at the levels:
\[
r_{\textrm{\$}\textrm{\euro}},\quad
r^{new}_{\textrm{\$}\textrm{\pounds}}=
{r_{\textrm{\$}\textrm{\euro}}\cdot
r_{\textrm{\euro}\textrm{\pounds}}},\quad
r_{\textrm{\euro}\textrm{\pounds}}.
\]

\emph{After the adjustment of the principal exchange rates
\eqref{prr3}, following the new information being revealed,  the
exchange rates become balanced, and this is the end of the
arbitrage evolution of an FX market with three currencies. Having
established the reasonably straightforward application of arbitrage
to three currencies, we now turn to investigation what happens when
the FX market contains four currencies and four currency traders.}

\section{Four Currencies}\label{S-4currencies}

Consider an FX market \$\textrm{\euro}\pounds\yen{} that involves
four currencies: Dollars (\$), Euros (\textrm{\euro}), Sterling
(\pounds) and Yen (\yen). This FX market involves six exchange
relationships:
\begin{alignat*}{3}
Dollar  &\rightleftarrows Euro,\quad & Dollar  &\rightleftarrows
Sterling,\quad & Dollar &\rightleftarrows Yen,\\
Euro   &\rightleftarrows Sterling, \quad & Euro &\rightleftarrows
Yen, &\quad Sterling   &\rightleftarrows Yen.
\end{alignat*}
The  exchange rates are:
\[
\begin{array}{llllll}
r_{\textrm{\$}\textrm{\euro}},& r_{\textrm{\$}\textrm{\pounds}},&r_{\textrm{\$}\textrm{\yen}}, &r_{\textrm{\euro}\textrm{\$}},& r_{\textrm{\euro}\textrm{\yen}},&r_{\textrm{\pounds}\textrm{\yen}},\\
r_{\textrm{\pounds}\textrm{\$}},& r_{\textrm{\pounds}\textrm{\euro}},&r_{\textrm{\pounds}\textrm{\yen}},& r_{\textrm{\yen}\textrm{\$}},& r_{\textrm{\yen}\textrm{\euro}},&r_{\textrm{\yen}\textrm{\pounds}}.
\end{array}
\]
The rates relating to the inverted arrows are reciprocal:
\begin{equation}\label{rec}
\begin{alignedat}{3}
r_{\textrm{\euro}\textrm{\$}}&=\frac{1}{r_{\textrm{\$}\textrm{\euro}}},\quad&
r_{\textrm{\pounds}\textrm{\$}}&=\frac{1}{r_{\textrm{\$}\textrm{\pounds}}},\quad&
r_{\textrm{\yen}\textrm{\$}}&=\frac{1}{r_{\textrm{\$}\textrm{\yen}}}, \\
r_{\textrm{\pounds}\textrm{\euro}}&=\frac{1}{r_{\textrm{\euro}\textrm{\pounds}}},\quad&
r_{\textrm{\yen}\textrm{\euro}}&=\frac{1}{r_{\textrm{\euro}\textrm{\yen}}},\quad &
r_{\textrm{\yen}\textrm{\pounds}}&=\frac{1}{r_{\textrm{\pounds}\textrm{\yen}}}.
\end{alignedat}
\end{equation}
Our market may be described by the ensemble of six {principal
exchange rates}
\begin{equation}\label{excr}
\calR=\left(r_{\textrm{\$}\textrm{\euro}},\
r_{\textrm{\$}\textrm{\pounds}},\
r_{\textrm{\$}\textrm{\yen}},\
r_{\textrm{\euro}\textrm{\pounds}},\
r_{\textrm{\euro}\textrm{\yen}},\
r_{\textrm{\pounds}\textrm{\yen}}\right)
\end{equation}
together with the reciprocal exchange rates \eqref{rec}.

The following characterisation of balanced, no-arbitrage, exchange
rates \eqref{excr}, that is the ensembles of  exchange rates such
that no trader could do better by trading indirectly, is
convenient.

\begin{proposition}\label{balp}
Ensemble \eqref{excr} of the principal exchange rates is balanced
if and only if the following relationships hold:
\begin{equation}
r_{\textrm{\euro}\textrm{\pounds}} =
\frac{r_{\textrm{\$}\textrm{\pounds}}}{r_{\textrm{\$}\textrm{\euro}}},\quad r_{\textrm{\euro}\textrm{\yen}} =
\frac{r_{\textrm{\$}\textrm{\yen}}}{r_{\textrm{\$}\textrm{\euro}}},
\label{invss}\quad r_{\textrm{\pounds}\textrm{\yen}} =
\frac{r_{\textrm{\$}\textrm{\yen}}}{r_{\textrm{\$}\textrm{\pounds}}}.
\end{equation}
\end{proposition}

\begin{proof}
This assertion can be proved by inspection.
\end{proof}

\section{Arbitrages}\label{S-arbitrages}

Let us suppose that initially each trader is aware only of the
three exchange rates involving his domestic currency. For instance,
the dollar trader knows only the rates
$r_{\textrm{\$}\textrm{\euro}}$, $r_{\textrm{\$}\textrm{\pounds}}$,
$r_{\textrm{\$}\textrm{\yen}}$.

\emph{We are interested in the case where the rates
$r_{\textrm{\$}\textrm{\euro}}$, $r_{\textrm{\$}\textrm{\pounds}}$,
$r_{\textrm{\$}\textrm{\yen}}$,
$r_{\textrm{\euro}\textrm{\pounds}}$,
$r_{\textrm{\euro}\textrm{\yen}}$,
$r_{\textrm{\pounds}\textrm{\yen}}$ are unbalanced}.

For instance, let us suppose that \FP can make a profit by first
exchanging one dollar for $r_{\textrm{\$}\textrm{\pounds}}$ units
of sterling, and then by exchanging this sterling  for euros. This
means that the product $r_{\textrm{\$}\textrm{\pounds}}\cdot
r_{\textrm{\pounds}\textrm{\euro}}$ is greater than
$r_{\textrm{\$}\textrm{\euro}}$:
\begin{equation}\label{unb}
r_{\textrm{\$}\textrm{\pounds}}\cdot
r_{\textrm{\pounds}\textrm{\euro}}>
r_{\textrm{\$}\textrm{\euro}}.
\end{equation}
Suppose that  \FP becomes aware of the rate
$r_{\textrm{\euro}\textrm{\pounds}}$, and, therefore, about the
inequality \eqref{unb}. The dollar trader then asks \AP to increase
the exchange rate $r_{\textrm{\$}\textrm{\euro}}$ to the new fairer
value
\[
r^{new}_{\textrm{\$}\textrm{\euro}}=
{r_{\textrm{\$}\textrm{\pounds}}}\cdot
r_{\textrm{\pounds}\textrm{\euro}}=
\frac{r_{\textrm{\$}\textrm{\pounds}}}{r_{\textrm{\euro}\textrm{\pounds}}}.
\]
Along with the adjustment of the exchange rate
$r_{\textrm{\$}\textrm{\euro}}$ the reciprocal rate
$r_{\textrm{\euro}\textrm{\$}}$ would be adjusted to
\[
r^{new}_{\textrm{\euro}\textrm{\$}}=\frac{1}{r^{new}_{\textrm{\$}\textrm{\euro}}}.
\]
We call this procedure \$\textrm{\euro}\pounds-\emph{arbitrage},
and we use the notation
$\calA_{\textrm{\$}\textrm{\euro}\textrm{\pounds}}$ to represent
it. We denote by
$\calR\calA_{\textrm{\$}\textrm{\euro}\textrm{\pounds}}$ the
ensemble of the new principal exchange rates:
\[
\calR^{new}=\calR\calA_{\textrm{\$}\textrm{\euro}\textrm{\pounds}}=
\left(r^{new}_{\textrm{\$}\textrm{\euro}},\
r_{\textrm{\$}\textrm{\pounds}},\
r_{\textrm{\$}\textrm{\yen}},\
r_{\textrm{\euro}\textrm{\pounds}},\
r_{\textrm{\euro}\textrm{\yen}},\
r_{\textrm{\pounds}\textrm{\yen}}
\right).
\]
We also use the notation
$\calR\calA_{\textrm{\$}\textrm{\euro}\textrm{\pounds}}$ in the
case where the inequality \eqref{unb} does not hold. In this case,
of course,
$\calR\calA_{\textrm{\$}\textrm{\euro}\textrm{\pounds}}=\calR$, and
we say that arbitrage
$\calA_{\textrm{\$}\textrm{\euro}\textrm{\pounds}}$ is \emph{not
active} in this case. This particular arbitrage is an example of
the 24 possible arbitrages listed in Table~\ref{tab1}. We will also
use, where convenient, the notation $\calA^{(n)}$ for the arbitrage
number $n$ from this table: for instance,
$\calA^{(1)}=\calA_{\textrm{\$}\textrm{\euro}\textrm{\pounds}}$.

\begin{table}[!htbp]
\caption{List of arbitrages}\label{tab1}
\begin{tabular}{llll}
Number&Arbitrage&Activation condition&Actions\\
\hline \\
1& $\calA_{\textrm{\$}\textrm{\euro}\textrm{\pounds}}$&
${r_{\textrm{\euro}\textrm{\pounds}}}>r_{\textrm{\$}\textrm{\euro}}\cdot
r_{\textrm{\$}\textrm{\pounds}}$ &
$r^{new}_{\textrm{\$}\textrm{\euro}}=
r_{\textrm{\$}\textrm{\pounds}}\cdot
r_{\textrm{\euro}\textrm{\pounds}}^{-1}$
\\

2& $\calA_{\textrm{\$}\textrm{\euro}\textrm{\yen}}$&
$r_{\textrm{\$}\textrm{\yen}}>r_{\textrm{\$}\textrm{\euro}}\cdot
r_{\textrm{\euro}\textrm{\yen}}$ &
$r^{new}_{\textrm{\$}\textrm{\euro}}= r_{\textrm{\$}\textrm{\yen}}\cdot r_{\textrm{\euro}\textrm{\yen}}^{-1}$\\

3& $\calA_{\textrm{\$}\textrm{\pounds}\textrm{\euro}}$&
$r_{\textrm{\$}\textrm{\euro}}\cdot
r_{\textrm{\euro}\textrm{\pounds}}>r_{\textrm{\$}\textrm{\pounds}}$
&
$r^{new}_{\textrm{\$}\textrm{\pounds}}= r_{\textrm{\$}\textrm{\euro}}\cdot r_{\textrm{\euro}\textrm{\pounds}}$\\

4& $\calA_{\textrm{\$}\textrm{\pounds}\textrm{\yen}}$&
$r_{\textrm{\$}\textrm{\yen}}>r_{\textrm{\$}\textrm{\pounds}}\cdot
r_{\textrm{\pounds}\textrm{\yen}}$ &
$r^{new}_{\textrm{\$}\textrm{\pounds}}= r_{\textrm{\$}\textrm{\yen}}\cdot r_{\textrm{\pounds}\textrm{\yen}}^{-1}$ \\

5& $\calA_{\textrm{\$}\textrm{\yen}\textrm{\euro}}$&
$r_{\textrm{\$}\textrm{\euro}}\cdot
r_{\textrm{\euro}\textrm{\yen}}>r_{\textrm{\$}\textrm{\yen}}$ &
$r^{new}_{\textrm{\$}\textrm{\yen}}= r_{\textrm{\$}\textrm{\euro}}\cdot r_{\textrm{\euro}\textrm{\yen}}$ \\

6& $\calA_{\textrm{\$}\textrm{\yen}\textrm{\pounds}}$&
$r_{\textrm{\$}\textrm{\pounds}}\cdot
r_{\textrm{\pounds}\textrm{\yen}}>r_{\textrm{\$}\textrm{\yen}}$ &
$r^{new}_{\textrm{\$}\textrm{\yen}}= r_{\textrm{\$}\textrm{\pounds}}\cdot r_{\textrm{\pounds}\textrm{\yen}}$ \\

7& $\calA_{\textrm{\euro}\textrm{\$}\textrm{\pounds}}$&
$r_{\textrm{\$}\textrm{\pounds}}<r_{\textrm{\$}\textrm{\euro}}\cdot
r_{\textrm{\euro}\textrm{\pounds}}$ &
$r^{new}_{\textrm{\$}\textrm{\euro}}= r_{\textrm{\$}\textrm{\pounds}}\cdot r_{\textrm{\euro}\textrm{\pounds}}^{-1}$ \\

8& $\calA_{\textrm{\euro}\textrm{\$}\textrm{\yen}}$&
$r_{\textrm{\$}\textrm{\yen}}<r_{\textrm{\$}\textrm{\euro}}\cdot
r_{\textrm{\euro}\textrm{\yen}}$ &
$r^{new}_{\textrm{\$}\textrm{\euro}}= r_{\textrm{\$}\textrm{\yen}}\cdot r_{\textrm{\euro}\textrm{\yen}}^{-1}$\\

9& $\calA_{\textrm{\euro}\textrm{\pounds}\textrm{\$}}$&
$r_{\textrm{\$}\textrm{\pounds}}>r_{\textrm{\euro}\textrm{\pounds}}\cdot
r_{\textrm{\$}\textrm{\euro}}$ &
$r^{new}_{\textrm{\euro}\textrm{\pounds}}= r_{\textrm{\$}\textrm{\pounds}}\cdot r_{\textrm{\$}\textrm{\euro}}^{-1}$ \\

10 & $\calA_{\textrm{\euro}\textrm{\pounds}\textrm{\yen}}$&
$r_{\textrm{\euro}\textrm{\yen}}>r_{\textrm{\euro}\textrm{\pounds}}\cdot
r_{\textrm{\pounds}\textrm{\yen}}$ &
$r^{new}_{\textrm{\euro}\textrm{\pounds}}= r_{\textrm{\euro}\textrm{\yen}}\cdot r_{\textrm{\pounds}\textrm{\yen}}^{-1}$ \\

11 & $\calA_{\textrm{\euro}\textrm{\yen}\textrm{\$}}$&
$r_{\textrm{\$}\textrm{\yen}}>r_{\textrm{\euro}\textrm{\yen}}\cdot
r_{\textrm{\$}\textrm{\euro}}$ &
$r^{new}_{\textrm{\euro}\textrm{\yen}}= r_{\textrm{\$}\textrm{\yen}}\cdot r_{\textrm{\$}\textrm{\euro}}^{-1}$ \\

12 & $\calA_{\textrm{\euro}\textrm{\yen}\textrm{\pounds}}$&
$r_{\textrm{\euro}\textrm{\pounds}}\cdot
r_{\textrm{\pounds}\textrm{\yen}}>r_{\textrm{\euro}\textrm{\yen}}$
&$r^{new}_{\textrm{\euro}\textrm{\yen}}= r_{\textrm{\euro}\textrm{\pounds}}\cdot r_{\textrm{\pounds}\textrm{\yen}}$ \\

13& $\calA_{\textrm{\pounds}\textrm{\$}\textrm{\euro}}$&
$r_{\textrm{\$}\textrm{\euro}}\cdot
r_{\textrm{\euro}\textrm{\pounds}}<r_{\textrm{\$}\textrm{\pounds}}$
&$r^{new}_{\textrm{\$}\textrm{\pounds}}= r_{\textrm{\$}\textrm{\euro}}\cdot r_{\textrm{\euro}\textrm{\pounds}}$\\

14& $\calA_{\textrm{\pounds}\textrm{\$}\textrm{\yen}}$&
$r_{\textrm{\$}\textrm{\yen}}<r_{\textrm{\$}\textrm{\pounds}}\cdot
r_{\textrm{\pounds}\textrm{\yen}}$ &
$r^{new}_{\textrm{\$}\textrm{\pounds}}= r_{\textrm{\$}\textrm{\yen}}\cdot r_{\textrm{\pounds}\textrm{\yen}}^{-1}$ \\

15& $\calA_{\textrm{\pounds}\textrm{\euro}\textrm{\$}}$&
$r_{\textrm{\$}\textrm{\pounds}}<r_{\textrm{\euro}\textrm{\pounds}}\cdot
r_{\textrm{\$}\textrm{\euro}}$ &
$r^{new}_{\textrm{\euro}\textrm{\pounds}}= r_{\textrm{\$}\textrm{\pounds}}\cdot r_{\textrm{\$}\textrm{\euro}}^{-1}$ \\

16 & $\calA_{\textrm{\pounds}\textrm{\euro}\textrm{\yen}}$&
$r_{\textrm{\euro}\textrm{\yen}}<r_{\textrm{\euro}\textrm{\pounds}}\cdot
r_{\textrm{\pounds}\textrm{\yen}}$ &
$r^{new}_{\textrm{\euro}\textrm{\pounds}}= r_{\textrm{\euro}\textrm{\yen}}\cdot r_{\textrm{\pounds}\textrm{\yen}}^{-1}$ \\

17& $\calA_{\textrm{\pounds}\textrm{\yen}\textrm{\$}}$&
$r_{\textrm{\$}\textrm{\yen}}>r_{\textrm{\pounds}\textrm{\yen}}\cdot
r_{\textrm{\$}\textrm{\pounds}}$ &
$r^{new}_{\textrm{\pounds}\textrm{\yen}}=
r_{\textrm{\$}\textrm{\yen}}\cdot
r_{\textrm{\$}\textrm{\pounds}}^{-1}$
\\

18& $\calA_{\textrm{\pounds}\textrm{\yen}\textrm{\euro}}$&
$r_{\textrm{\euro}\textrm{\yen}}>r_{\textrm{\pounds}\textrm{\yen}}\cdot
r_{\textrm{\euro}\textrm{\pounds}}$ &
$r^{new}_{\textrm{\pounds}\textrm{\yen}}= r_{\textrm{\euro}\textrm{\yen}}\cdot r_{\textrm{\euro}\textrm{\pounds}}^{-1}$\\

19& $\calA_{\textrm{\yen}\textrm{\$}\textrm{\euro}}$&
$r_{\textrm{\$}\textrm{\euro}}\cdot
r_{\textrm{\euro}\textrm{\yen}}<r_{\textrm{\$}\textrm{\yen}}$ &
$r^{new}_{\textrm{\$}\textrm{\yen}}= r_{\textrm{\$}\textrm{\euro}}\cdot r_{\textrm{\euro}\textrm{\yen}}$ \\

20& $\calA_{\textrm{\yen}\textrm{\$}\textrm{\pounds}}$&
$r_{\textrm{\$}\textrm{\pounds}}\cdot
r_{\textrm{\pounds}\textrm{\yen}}<r_{\textrm{\$}\textrm{\yen}}$ &
$r^{new}_{\textrm{\$}\textrm{\yen}}= r_{\textrm{\$}\textrm{\pounds}}\cdot r_{\textrm{\pounds}\textrm{\yen}}$ \\

21 & $\calA_{\textrm{\yen}\textrm{\euro}\textrm{\$}}$&
$r_{\textrm{\$}\textrm{\yen}}<r_{\textrm{\euro}\textrm{\yen}}\cdot
r_{\textrm{\$}\textrm{\euro}}$ &
$r^{new}_{\textrm{\euro}\textrm{\yen}}= r_{\textrm{\$}\textrm{\yen}}\cdot r_{\textrm{\$}\textrm{\euro}}^{-1}$ \\

22 & $\calA_{\textrm{\yen}\textrm{\euro}\textrm{\pounds}}$&
$r_{\textrm{\euro}\textrm{\pounds}}\cdot
r_{\textrm{\pounds}\textrm{\yen}}<r_{\textrm{\euro}\textrm{\yen}}$
&$r^{new}_{\textrm{\euro}\textrm{\yen}}= r_{\textrm{\euro}\textrm{\pounds}}\cdot r_{\textrm{\pounds}\textrm{\yen}}$ \\

23& $\calA_{\textrm{\yen}\textrm{\pounds}\textrm{\$}}$&
$r_{\textrm{\$}\textrm{\yen}}<r_{\textrm{\pounds}\textrm{\yen}}\cdot
r_{\textrm{\$}\textrm{\pounds}}$ &
$r^{new}_{\textrm{\pounds}\textrm{\yen}}=
r_{\textrm{\$}\textrm{\yen}}\cdot
r_{\textrm{\$}\textrm{\pounds}}^{-1}$
\\

24& $\calA_{\textrm{\yen}\textrm{\pounds}\textrm{\euro}}$&
$r_{\textrm{\euro}\textrm{\yen}}<r_{\textrm{\pounds}\textrm{\yen}}\cdot
r_{\textrm{\euro}\textrm{\pounds}}$ &
$r^{new}_{\textrm{\pounds}\textrm{\yen}}= r_{\textrm{\euro}\textrm{\yen}}\cdot r_{\textrm{\euro}\textrm{\pounds}}^{-1}$\\
\end{tabular}
\end{table}

\emph{The principal distinction of the  FX market with four
currencies from that with only three currencies is that applying a
single arbitrage operation does not bring the FX market to a
balance in which no arbitrage opportunities exist, and in which the
law of one price holds.}

\section{Main Results}\label{mrSS}

One can apply arbitrages from Table~\ref{tab1} sequentially in any
order and to any initial exchange rates $\calR$. The situation that
we have in mind is the following. Suppose that there exists an
\emph{Arbiter} who knows current ensemble $\calR$ of exchange
rates. This \emph{Arbiter} could provide information to the FX
traders in any order he wants, thus activating the \emph{chain} (or
\emph{superposition}) of corresponding arbitrages. The principal
question is:
\begin{question}\label{que1}
How powerful is the \emph{Arbiter}?
\end{question}

The short answer is: the \emph{Arbiter is surprisingly powerful.}

Let us explain at a more formal level what we mean.

For a finite chain of arbitrages $ {\bA} = \calA_1 \cdots \calA_n$,
and for a  given ensemble $\calR$ of initial exchange rates, we
denote the resulting ensemble of principal exchange rates as
\begin{equation}\label{arbsec}
\calR{\bA}=\calR\calA_1 \cdots \calA_n
\end{equation}

If $\calR$ is balanced, then $\calR{\calA}=\calR$ for any
individual arbitrage, and therefore $\calR{\bA}=\calR$ for any
chain \eqref{arbsec}. If, on the contrary, $\calR$ is not balanced,
then different arbitrage chains \eqref{arbsec} could result in
different balanced or unbalanced ensembles of principal exchange
rates. Denote by $S(\calR)$ the collection of the sets $\calR{\bA}$
related to all possible chains \eqref{arbsec}. Denote also by
$S^{bal}(\calR)$ the subset of $S(\calR)$, that includes only
balanced exchange rates ensembles. Our principal observation is the
following:

\emph{For a typical unbalanced exchange rate ensemble $\calR$, the
set $S^{bal}(\calR)$ is unexpectedly rich; therefore the Arbiter,
who prescribes a particular sequence of arbitrages, is an
unexpectedly powerful figure.}

To avoid cumbersome notation and technical details when providing a
rigorous formulation of this observation, we concentrate on the
simplest initial ensemble. Let us consider the ensemble
\begin{equation}\label{dist}
\bar\calR_{\alpha}=\left(\alpha \cdot \bar{r}_{\textrm{\$}\textrm{\euro}},\
\bar{r}_{\textrm{\$}\textrm{\pounds}},\ \bar{r}_{\textrm{\$}\textrm{\yen}},\ \bar{r}_{\textrm{\euro}\textrm{\pounds}},\ \bar{r}_{\textrm{\euro}\textrm{\yen}},\ \bar{r}_{\textrm{\pounds}\textrm{\yen}}
\right),
\end{equation}
where $\alpha>0, \alpha \not= 1$ and $\bar\calR$ is
a given balanced ensemble of principal exchange
rates. The ensemble  \eqref{dist} is not balanced.
The ensemble \eqref{dist} may have emerged as
follows. Let us suppose that the underlying balanced
rates
\begin{equation}\label{distbar}
\bar\calR=\left(\bar{r}_{\textrm{\$}\textrm{\euro}},\ \bar{r}_{\textrm{\$}\textrm{\pounds}},\
\bar{r}_{\textrm{\$}\textrm{\yen}},\ \bar{r}_{\textrm{\euro}\textrm{\pounds}},\ \bar{r}_{\textrm{\euro}\textrm{\yen}},\ \bar{r}_{\textrm{\pounds}\textrm{\yen}}
\right)
\end{equation}
had been in operation up to a certain reference time moment $0$. At
this moment {\FP} has decided to increase his price for euros by a
factor $\alpha>1$. A natural respecification of Question \ref{que1}
is the following:
\begin{question}\label{que2}
To which balanced exchange rates can the Arbiter now bring the
foreign exchange market?
\end{question}

The possible general structure of elements from the corresponding
sets $S(\bar\calR_{\alpha})$  and $S^{bal}(\bar\calR_{\alpha})$ is
easy to describe. To this end we denote by $T_{\alpha}(\bar\calR)$
the collection of all sextuples of the form
\begin{equation}\label{prod}
\left(\alpha^{n_{1}} \cdot \bar{r}_{\textrm{\$}\textrm{\euro}},\
\alpha^{n_{2}} \cdot \bar{r}_{\textrm{\$}\textrm{\pounds}},\
\alpha^{n_{3}} \cdot \bar{r}_{\textrm{\$}\textrm{\yen}},\
\alpha^{n_{4}}\cdot \bar{r}_{\textrm{\euro}\textrm{\pounds}},\
\alpha^{n_{5}}\cdot \bar{r}_{\textrm{\euro}\textrm{\yen}} ,\
\alpha^{n_{6}}\cdot \bar{r}_{\textrm{\pounds}\textrm{\yen}}
\right),
\end{equation}
where $n_{i}$ are integer numbers (positive, negative or zero). We
also denote by $T^{bal}_{\alpha}$ the subset of elements of
$T_{\alpha}$, which satisfy the relationships
\[
n_{4} =  n_{2}-n_{1},\quad n_{5} =  n_{3}-n_{1},\quad n_{6} =
n_{3}-n_{2}.
\]

\begin{proposition}\label{rep1P}
The following inclusions hold:
\begin{align}\label{talp}
S(\calR_{\alpha})&\subset T_{\alpha}(\bar\calR),\\
\label{palpb}
S^{bal}(\bar\calR_{\alpha})&\subset T^{bal}_{\alpha}(\bar\calR).
\end{align}
\end{proposition}

\begin{proof}
The ensemble \eqref{distbar}  belongs to $T$. To verify
\eqref{talp} we show that the set $T_{\alpha}$ is invariant with
respect to each arbitrage $\calA$ from Table~\ref{tab1}. This
statement can be checked by inspection. Let us, for instance, apply
to a sextuple \eqref{prod} the first arbitrage
$\calA_{\textrm{\$}\textrm{\euro}\textrm{\pounds}}$. Then, by
definition, either this arbitrage is inactive, or it changes the
first component $\alpha^{n_{1}} \cdot
\bar{r}_{\textrm{\$}\textrm{\euro}} $ of \eqref{prod} to the new
value
\begin{equation}\label{bb}
r^{new}_{\textrm{\$}\textrm{\euro}}= \frac{\alpha^{n_{2}} \cdot \bar{r}_{\textrm{\$}\textrm{\pounds}}}
{\alpha^{n_{4}}\cdot \bar{r}_{\textrm{\euro}\textrm{\pounds}}}= \alpha^{n_{2}-n_{4}}\cdot
\frac{\bar{r}_{\textrm{\$}\textrm{\pounds}}}{ \bar{r}_{\textrm{\euro}\textrm{\pounds}}}.
\end{equation}
However, the ensemble ${\bar\calR}$ is balanced, and,  by the first
equation \eqref{invss},
$\frac{\bar{r}_{\textrm{\$}\textrm{\pounds}}}{\bar{r}_{\textrm{\euro}\textrm{\pounds}}}=
\bar{r}_{\textrm{\$}\textrm{\euro}}$.
Therefore, \eqref{bb} implies that the ensemble
$\bar\calR\calA_{\textrm{\$}\textrm{\euro}\textrm{\pounds}}$ also
may be represented in the form \eqref{prod}. We have proved the
first part of the proposition, related to the set
$S(\bar\calR_{\alpha})$. The inclusion \eqref{palpb} follows now
from Proposition~\ref{balp}.
\end{proof}

Proposition~\ref{rep1P} in no way answers Question~\ref{que2}. This
proposition, however, allows us to reformulate this question in a
more constructive form:

\begin{question}\label{que3}
How big is the set $S^{bal}(\bar\calR_{\alpha})$, compared with the
collection  $T^{bal}_{\alpha}(\bar\calR)$  of all elements that
satisfy the restrictions imposed by Proposition~\ref{rep1P}?
\end{question}

The naive expectation would be that the set
$S^{bal}(\bar\calR_{\alpha})$ is finite and, at least for values of
$\alpha$ close to 1, that all elements of
$S^{bal}(\bar\calR_{\alpha})$ are close to $\bar\calR$. However,the
following statement, describing an unexpected feature of the power
of the \emph{Arbiter}, is true.

\begin{theorem}\label{arbH}
The set  $S^{bal}(\bar\calR_{\alpha})$ coincides with
$T^{bal}_{\alpha}(\bar\calR)$:
\begin{equation}\label{palpe}
S^{bal}(\bar\calR_{\alpha})= T^{bal}_{\alpha}(\bar\calR).
\end{equation}
Moreover each balanced ensemble \eqref{prod} may be achieved via a
chain of arbitrage operations no longer than
\begin{equation}
\label{estN} N(n_{1},n_{2},n_{3})=3(|n_{1}-1|+|n_{2}|+|n_{3}|)+ 3.
\end{equation}
\end{theorem}

Loosely speaking, this theorem means that the \emph{Arbiter} is
extremely  powerful.  An assertion similar to Theorem~\ref{arbH}
was formulated as a hypothesis in \cite{KozCalPok:ArXiv10}. We
describe the  algorithms corresponding to this theorem in the next
section.

The following assertion certifies that the estimate \eqref{estN}
from Theorem \ref{arbH} is pretty close to the optimal.

\begin{proposition}\label{VictorP}
The inequalities
\[
\label{dist2}
 |n_1-n_2+n_4|,~|n_1-n_3+n_5)|,~|n_2-n_3+n_6|\le 1
 \]
 hold for any $\calR\in S(\bar\calR_{\alpha})$.
 Here $n_i$ are the integers from representation \eqref{prod} of
 $\calR$.
 \end{proposition}

\begin{proof} This assertion is a special case of
Lemma~\ref{specialC} which will be considered below.
\end{proof}

Note that the set $S(\bar\calR_{\alpha})$ is, in contrast to
\eqref{palpe}, much smaller than the totality
$T_{\alpha}(\bar\calR)$ of all ensembles of the form \eqref{prod}.
In particular, the following assertion holds:

\begin{proposition}\label{boundP}
Let ${\bA}$ denote a chain of arbitrages  of length $N$, and
$\calR=\bar\calR_{\alpha}{\bA}$.  Then
$3(|n_{1}-1|+|n_{2}|+|n_{3}|)\le N+8$, where $n_1,n_2,n_3$ are the
integers from the representation \eqref{prod} of $\calR$.
\end{proposition}

Let us consider an infinite arbitrage chain:
\begin{equation}\label{perse}
{\bA}={\calA_{1}}{\calA_{2}}{\calA_{3}} \cdots  {\calA_{n}} \dotsm .
\end{equation}
This chain is periodic with minimal period $p$ if
${\calA_{n}}={\calA_{n+p}}$ for $n=1,2,\ldots $, and $p$ is the
minimal positive integer with this property. Various periodic
chains of arbitrage play a  special role in context of this
article, and we summairise below some interesting features of such
periodic arbitrage chains. For a periodic chain \eqref{perse} and
for an initial (unbalanced) exchange rate ensemble $\calR_{0}$ we
consider the sequence
\begin{equation}\label{pero}
{\calR_{0}},{\calR_{1}},{\calR_{2}}, \ldots , {\calR_{n}}, \ldots
\end{equation}
defined by $ \calR_{n}=\calR_{n-1}\calA_{n}$, $n=1,2,\dotsc$.

\begin{proposition} \label{perpro}
Either (i) the sequence \eqref{pero} is periodic for $n\ge 36p;$
or (ii) this sequence is diverging:
at least one of the following six relationships hold:
\[
{ r_{\textrm{\$}\textrm{\euro}}}_{n}\to 0,\
{\ r_{\textrm{\$}\textrm{\pounds}}}_{n}\to 0,\
{r_{\textrm{\$}\textrm{\yen}}}_{n}\to 0,\
{ r_{\textrm{\$}\textrm{\euro}}}_{n}\to \infty,\
{\ r_{\textrm{\$}\textrm{\pounds}}}_{n}\to \infty,\
{r_{\textrm{\$}\textrm{\yen}}}_{n}\to \infty.
\]
Moreover, in Case (i) the minimal period of the sequence is a
divisor of $24p$; in Case (ii) there exist a divisor $q$ of $24p$
and factors
$\gamma_{\textrm{\$}\textrm{\euro}},\ldots,\gamma_{\textrm{\pounds}\textrm{\yen}}$
such that the relationships ${r_{\textrm{\$}\textrm{\euro}}}_{n+q}=
\gamma_{\textrm{\$}\textrm{\euro}}
{r_{\textrm{\$}\textrm{\euro}}}_{n},\ \ldots,\
{r_{\textrm{\pounds}\textrm{\yen}}}_{n+q}=
\gamma_{\textrm{\pounds}\textrm{\yen}}
{r_{\textrm{\$}\textrm{\euro}}}_{n} $ hold for $n\ge 36p$.
\end{proposition}

\begin{proof}
This statement follows from Lemmas~\ref{comp} and~\ref{trans}.
\end{proof}

To conclude this discussion, we note one more unexpected feature of
periodic chains of arbitrage. A chain  \eqref{perse} is
\emph{regular for the initial ensemble $\calR_0$} if this chain
includes all 24 arbitrages, and each arbitrage is active infinitely
many times while generating the sequence \eqref{pero}. By analogy
with typical results from the desynchronised systems theory, one
could expect a regular chain of arbitrage elements of the
corresponding sequence  \eqref{pero} should be balanced for
sufficiently large $n$.  However, this is not the case: the
sequences  \eqref{pero} may be both periodic (after some transient
period) or diverging.

As an instructive example consider the 24-periodic chain
${\bA}_{*}$  which is defined by the following equations:
\begin{alignat*}{4}
\calA_{1}&=\calA^{(15)},\quad&\calA_{2}&=\calA^{(10)},\quad&\calA_{3}&=\calA^{(3)},\quad&\calA_{4}&=\calA^{(21)},\\
\calA_{5}&=\calA^{(11)},\quad&\calA_{6}&=\calA^{(8)},\quad&\calA_{7}&=\calA^{(24)},\quad&\calA_{8}&=\calA^{(17)},\\
\calA_{9}&=\calA^{(6)},\quad&\calA_{10}&=\calA^{(9)},\quad&\calA_{11}&=\calA^{(16)},\quad&\calA_{12}&=\calA^{(13)},\\
\calA_{13}&=\calA^{(12)},\quad&\calA_{14}&=\calA^{(22)},\quad&\calA_{15}&=\calA^{(14)},\quad&\calA_{16}&=\calA^{(18)},\\
\calA_{17}&=\calA^{(23)},\quad&\calA_{18}&=\calA^{(15)},\quad&\calA_{19}&=\calA^{(5)},\quad&\calA_{20}&=\calA^{(7)},\\
\calA_{21}&=\calA^{(4)},\quad&\calA_{22}&=\calA^{(19)},\quad&\calA_{23}&=\calA^{(1)},\quad&\calA_{24}&=\calA^{(5)}.
\end{alignat*}

\begin{proposition} \label{32}
For the initial ensemble $\calR_0= \bar\calR_{\alpha}$  the
corresponding sequence \eqref{pero} is periodic with minimal period
24, and all arbitrages from ${\bA}_{*}$ are active.
\end{proposition}

\begin{proof}
By inspection.
\end{proof}

This proposition demonstrates that arbitrage operation chains may
display periodicity and no necessary convergence on a cross
exchange rate law of one price. See Figs.~\ref{GraphLeha1_24},
\ref{GraphLeha3} and formula \eqref{geomrout} below for an
explanation of the geometrical meaning of the arbitrage chain
${\bA}_{*}$.

\section{The Basic Algorithm}\label{simpleASS}

Introduce the following chains of arbitrages of length $3$:

{\small\begin{alignat*}{3}
{\bA}_{+}^{(1)}&=\calA^{(21)}\calA^{(16)}\calA^{(1)},\quad&
{\bA}_{+}^{(2)}&=\calA^{(3)}\calA^{(17)}\calA^{(10)},\quad&
{\bA}_{+}^{(3)}&=\calA^{(5)}\calA^{(18)}\calA^{(12)},\\
{\bA}_{-}^{(1)}&=\calA^{(8)}\calA^{(9)}\calA^{(11)},\quad&
{\bA}_{-}^{(2)}&=\calA^{(15)}\calA^{(18)}\calA^{(14)},\quad&
{\bA}_{-}^{(3)}&=\calA^{(21)}\calA^{(23)}\calA^{(20)}.
\end{alignat*}}

\noindent It is convenient to define the mapping $\sigma(n)$ which
corresponds to a non-negative integer $n$ by the symbol ``$+$'',
and by the symbol ``$-$'' for a negative integer.
\begin{proposition}\label{algP}
The chain
\begin{equation} \label{alg0e} {\bA}(n_{1},n_{2},n_{3})=
 {\left({\bA}_{\sigma( n_{3})}^{(3)}\right)}^{|n_{3}|}
 {\left({\bA}_{\sigma (n_{2})}^{(2)}\right)}^{|n_{2}|}
 \calA^{(15)}\calA^{(18)}
 {\left({\bA}_{\sigma (n_{1})}^{(1)}\right)}^{|n_{1}-1|}
 \calA^{(5)}
\end{equation}
satisfies Theorem~\ref{arbH}: the ensemble
$\bar\calR_{\alpha}{\bA}(n_{1},n_{2},n_{3})$ coincides with
\[
\left(\alpha^{n_{1}} \cdot \bar{r}_{\textrm{\$}\textrm{\euro}} ,\
\alpha^{n_{2}} \cdot \bar{r}_{\textrm{\$}\textrm{\pounds}},\
\alpha^{n_{3}} \cdot \bar{r}_{\textrm{\$}\textrm{\yen}} ,\
\alpha^{n_{1}-n_{2}}\cdot \bar{r}_{\textrm{\euro}\textrm{\pounds}},\
\alpha^{n_{1}-n_{3}}\cdot \bar{r}_{\textrm{\euro}\textrm{\yen}} ,\
\alpha^{n_{2}-n_{3}}\cdot \bar{r}_{\textrm{\pounds}\textrm{\yen}}
\right),
\]
and the length  $N$ of the chain \eqref{alg0e} satisfies
$
N\le 3(|n_{1}-1|+|n_{2}|+|n_{3}|)+ 3.
$
\end{proposition}

The legitimacy of this algorithm may be verified by induction.
However a simple geometric proof is much more instructive. This
proof will be given later on. This chain is not always the
shortest: for instance, in the case $n_{1}= n_{2}=n_{3}=0$ the
shortest chain ${\bA}$ is of length one: ${\bA}=\calA_{7}$.

\section{General case}\label{S-gencase}
\subsection{Direct Generalisation}\label{directSS}
We begin with the following comment. The ensemble \eqref{dist} is
the first item in the list
\begin{equation}\label{list}
\begin{split}
\bar\calR^{1}_{\alpha}&=\left(\alpha \cdot \bar{r}_{\textrm{\$}\textrm{\euro}},\
\bar{r}_{\textrm{\$}\textrm{\pounds}},\ \bar{r}_{\textrm{\$}\textrm{\yen}},\ \bar{r}_{\textrm{\euro}\textrm{\pounds}},\ \bar{r}_{\textrm{\euro}\textrm{\yen}},\ \bar{r}_{\textrm{\pounds}\textrm{\yen}}\right),\\
\bar\calR^{2}_{\alpha}&=\left(\bar{r}_{\textrm{\$}\textrm{\euro}},\
\alpha\cdot\bar{r}_{\textrm{\$}\textrm{\pounds}},\ \bar{r}_{\textrm{\$}\textrm{\yen}},\ \bar{r}_{\textrm{\euro}\textrm{\pounds}},\ \bar{r}_{\textrm{\euro}\textrm{\yen}},\ \bar{r}_{\textrm{\pounds}\textrm{\yen}}\right),\\
\bar\calR^{3}_{\alpha}&=\left(\bar{r}_{\textrm{\$}\textrm{\euro}},\
\bar{r}_{\textrm{\$}\textrm{\pounds}},\ \alpha\cdot\bar{r}_{\textrm{\$}\textrm{\yen}},\ \bar{r}_{\textrm{\euro}\textrm{\pounds}},\ \bar{r}_{\textrm{\euro}\textrm{\yen}},\ \bar{r}_{\textrm{\pounds}\textrm{\yen}}\right),\\
\bar\calR^{4}_{\alpha}&=\left(\bar{r}_{\textrm{\$}\textrm{\euro}},\
\bar{r}_{\textrm{\$}\textrm{\pounds}},\ \bar{r}_{\textrm{\$}\textrm{\yen}},\ \alpha\cdot\bar{r}_{\textrm{\euro}\textrm{\pounds}},\ \bar{r}_{\textrm{\euro}\textrm{\yen}},\ \bar{r}_{\textrm{\pounds}\textrm{\yen}}\right),\\
\bar\calR^{5}_{\alpha}&=\left(\bar{r}_{\textrm{\$}\textrm{\euro}},\
\bar{r}_{\textrm{\$}\textrm{\pounds}},\ \bar{r}_{\textrm{\$}\textrm{\yen}},\ \bar{r}_{\textrm{\euro}\textrm{\pounds}},\ \alpha\cdot\bar{r}_{\textrm{\euro}\textrm{\yen}},\ \bar{r}_{\textrm{\pounds}\textrm{\yen}}\right),\\
\bar\calR^{6}_{\alpha}&=\left(\bar{r}_{\textrm{\$}\textrm{\euro}},\
\bar{r}_{\textrm{\$}\textrm{\pounds}},\ \bar{r}_{\textrm{\$}\textrm{\yen}},\ \bar{r}_{\textrm{\euro}\textrm{\pounds}},\ \bar{r}_{\textrm{\euro}\textrm{\yen}},\ \alpha\cdot\bar{r}_{\textrm{\pounds}\textrm{\yen}}
\right).
\end{split}
\end{equation}
A natural ``relabelling'' procedure confirms that the main results
described in Section~\ref{mrSS} hold without any changes for first
initial ensemble from the list \eqref{list}. In particular, Theorem
\ref{arbH} implies
\begin{corollary}\label{arbAH}
The equality $S^{bal}(\bar\calR^{i}_{\alpha})=
T^{bal}_{\alpha}(\bar\calR)$ holds for $i=2,3$.  Moreover each
balanced ensemble \eqref{prod} may be achieved via a chain of
arbitrage operations no longer than $N^{i}(n_{1},n_{2},n_{3})$,
where
\begin{align*}
N^{2}(n_{1},n_{2},n_{3})&=3(|n_{1}|+|n_{2}-1|+|n_{3}|)+3,\\
N^{3}(n_{1},n_{2},n_{3})&=3(|n_{1}|+|n_{2}|+|n_{3}-1|)+3,\\
\end{align*}
\end{corollary}

To describe the corresponding algorithms we introduce the auxiliary
chains {\small\begin{alignat*}{3}
{\tbA}_{+}^{(1)}&=\calA^{(1)}\calA^{(21)}\calA^{(16)},\quad&
{\tbA}_{+}^{(2)}&=\calA^{(13)}\calA^{(23)}\calA^{(16)},\quad&
{\tbA}_{+}^{(3)}&=\calA^{(24)}\calA^{(12)}\calA^{(19)}.\\
{\tbA}_{-}^{(1)}&=\calA^{(9)}\calA^{(11)}\calA^{(8)},\quad&
{\tbA}_{-}^{(2)}&=\calA^{(9)}\calA^{(34)}\calA^{(4)},\quad&
{\tbA}_{-}^{(3)}&=\calA^{(6)}\calA^{(11)}\calA^{(17)};\\
{\ttbA}_{+}^{(1)}&=\calA^{(18)}\calA^{(12)}\calA^{(5)},\quad&
{\ttbA}_{+}^{(2)}&=\calA^{(23)}\calA^{(16)}\calA^{(13)},\quad&
{\ttbA}_{+}^{(3)}&=\calA^{(18)}\calA^{(12)}\calA^{(5)}.\\
{\ttbA}_{-}^{(1)}&=\calA^{(20)}\calA^{(21)}\calA^{(23)},\quad&
{\ttbA}_{-}^{(2)}&=\calA^{(4)}\calA^{(9)}\calA^{(24)},\quad&
{\ttbA}_{-}^{(3)}&=\calA^{(20)}\calA^{(21)}\calA^{(23)}.
\end{alignat*}}

The equation \eqref{alg0e} can be modified to the form
{\small\begin{alignat*}{1}
 {\bA}_{2}(n_{1},n_{2},n_{3}) &=
{\left({\tbA}_{\sigma( n_{1})}^{(1)}\right)}^{|n_{1}|}
\calA^{(24)}\calA^{(12)} {\left({\tbA}_{\sigma
(n_{3})}^{(3)}\right)}^{|n_{3}|} {\left({\tbA}_{\sigma
(n_{2})}^{(2)}\right)}^{|n_{2}-1|} \calA^{(1)}
\end{alignat*}}
for $i=2$, and to the form
{\small\begin{alignat*}{1}
 {\bA}_{3}(n_{1},n_{2},n_{3}) %&={\bA}(n_{1}+1,n_{2}-1,n_{3})\\
 &= {\left({\ttbA}_{\sigma (n_{2})}^{(2)}\right)}^{|n_{2}|}
 {\left({\ttbA}_{\sigma (n_{1})}^{(1)}\right)}^{|n_{1}|}
 \calA^{(12)}\calA^{(10)}
 {\left({\ttbA}_{\sigma (n_{3})}^{(3)}\right)}^{|n_{3}-1|}
 \calA^{(3)}
\end{alignat*}}
for $i=3$.

Let us turn to the initial ensembles $\bar\calR^{i}_{\alpha}$,
$i=4,5,6$.

\begin{proposition}\label{arbA2H}
The equality $ S^{bal}(\bar\calR^{i}_{\alpha})=
T^{bal}_{\alpha}(\bar\calR)$ holds for $i=4,5,6$.  Moreover each
balanced ensemble \eqref{prod} may be achieved via a chain of
arbitrage no longer than $N^{i}(n_{1},n_{2},n_{3})$, where
\[
N^{4,5,6}(n_{1},n_{2},n_{3})=3(|n_{1}|+|n_{2}|+|n_{3}|)+4.
\]
\end{proposition}

The corresponding chains ${\bA}_{4}(n_{1},n_{2},n_{3})$, $i=4,5,6$,
may be defined by the following equations:
\begin{multline*}
{\bA}_{4}(n_{1},n_{2},n_{3})=\calA^{(12)}{\bA}(n_{1}+1,n_{2},n_{3})\\
 =\calA^{(12)}{\left({\bA}_{\sigma( n_{3})}^{(3)}\right)}^{|n_{3}|}
   {\left({\bA}_{\sigma (n_{2})}^{(2)}\right)}^{|n_{2}|}
 \calA^{(15)}\calA^{(18)}
  {\left({\bA}_{\sigma( n_{1})}^{(1)}\right)}^{|n_{1}|}
 \calA^{(5)},
 \end{multline*}
\begin{multline*}
{\bA}_{5}(n_{1},n_{2},n_{3})=\calA^{(16)}{\bA}(n_{1}+1,n_{2},n_{3})\\
 =\calA^{(16)}{\left({\bA}_{\sigma (n_{3})}^{(3)}\right)}^{|n_{3}|}
  {\left({\bA}_{\sigma (n_{2})}^{(2)}\right)}^{|n_{2}|}
 \calA^{(15)}\calA^{(18)}
  {\left({\bA}_{\sigma (n_{1})}^{(1)}\right)}^{|n_{1}|}
 \calA^{(3)},
 \end{multline*}
 \begin{multline*}
{\bA}_{6}(n_{1},n_{2},n_{3})=\calA^{(16)}{\bA}_{2}(n_{1}+1,n_{2},n_{3})\\
 =\calA^{(10)} {\left({\tbA}_{\sigma (n_{1})}^{(1)}\right)}^{|n_{1}|}
 \calA^{(24)}\calA^{(12)}
 {\left({\tbA}_{\sigma (n_{3})}^{(3)}\right)}^{|n_{3}|}
 {\left({\tbA}_{\sigma (n_{2})}^{(2)}\right)}^{|n_{2}-1|}
 \calA^{(1)}.
 \end{multline*}

\begin{proof}
This assertion may be proved analogously to Theorem~\ref{arbH}.
\end{proof}

\subsection{Arbitrage Discrepancies\label{DiscSS}}
To formulate further generalisations we need an additional notion.
To each ensemble $\calR= \left(r_{\textrm{\$}\textrm{\euro}},
r_{\textrm{\$}\textrm{\pounds}}, r_{\textrm{\$}\textrm{\yen}},
r_{\textrm{\euro}\textrm{\pounds}},
r_{\textrm{\euro}\textrm{\yen}},
r_{\textrm{\pounds}\textrm{\yen}}\right)$
we attach an {\em
arbitrage discrepancies ensemble,} using the relationships for
balanced principal exchange rates given in \eqref{invss} above
\[
\calD(\calR)=\left(d_{\textrm{\euro}\textrm{\pounds}}(\calR),d_{\textrm{\euro}\textrm{\yen}}(\calR),d_{\textrm{\pounds}\textrm{\yen}}(\calR)\right)
\]
as follows:
\begin{equation}\label{discrep}
\begin{split}
d_{\textrm{\euro}\textrm{\pounds}}(\calR)&=
\log r_{\textrm{\euro}\textrm{\pounds}}-\log r_{\textrm{\$}\textrm{\pounds}}+\log r_{\textrm{\$}\textrm{\euro}},\\
d_{\textrm{\euro}\textrm{\yen}}(\calR)&=
\log r_{\textrm{\euro}\textrm{\yen}}-\log r_{\textrm{\$}\textrm{\yen}}+\log r_{\textrm{\$}\textrm{\euro}},\\
d_{\textrm{\pounds}\textrm{\yen}}(\calR)&= \log
r_{\textrm{\pounds}\textrm{\yen}}-\log
r_{\textrm{\$}\textrm{\yen}}+\log r_{\textrm{\$}\textrm{\pounds}}.
\end{split}
\end{equation}
For instance
\begin{equation}\label{fori}
\begin{alignedat}{3}
\calD(\bar\calR^{1}_{\alpha})&=a(1,1,0),\quad&
\calD(\bar\calR^{2}_{\alpha})&=a(-1,0,1),\quad&
\calD(\bar\calR^{3}_{\alpha})&=a(0,-1,-1),\\
\calD(\bar\calR^{4}_{\alpha})&=a(1,0,0),\quad&
\calD(\bar\calR^{5}_{\alpha})&=a(0,1,0),\quad&
\calD(\bar\calR^{6}_{\alpha})&=a(0,0,1),
\end{alignedat}
\end{equation}
where $a=\log \alpha$.

\begin{proposition}\label{disc0P}
The ensemble $\calR$ is balanced, if and only if $\calD(\calR)=0$.
\end{proposition}

\begin{proof}
Follows from  Proposition~\ref{balp} and equations \eqref{discrep}.
\end{proof}

\subsection{Case A\label{caseA}}
The case where two of the discrepancies \eqref{discrep} are equal
to zero was implicitly considered in Section~\ref{directSS}: see
the second line in \eqref{fori} and Proposition~\ref{arbA2H}.

\subsection{Case B\label{caseB}}
Consider now the case when one of the discrepancies in
\eqref{discrep} is equal to zero, while two others are not. We will
be particularly interested in the situation where two nonzero
discrepancies are different. This situation may have emerged, for
instance, as follows. Let us suppose that the underlying balanced
rates \eqref{distbar} had been in operation up to a certain
reference time moment $0$. At this moment the Euro trader has
decided to change two of  three his rates, namely
$r_{\textrm{\euro}\textrm{\pounds}}$ and
$r_{\textrm{\euro}\textrm{\yen}}$, by different factors $\alpha$
and $\beta$. Then at this moment the two discrepancies would
acquire different non-zero values, while the third discrepancy
remains equal to zero.

Suppose, for example that $d_{\textrm{\pounds}\textrm{\yen}}=0$,
while $d_{\textrm{\euro}\textrm{\pounds}},
d_{\textrm{\euro}\textrm{\yen}}\not=0$. We introduce the ratio
\begin{equation}\label{ratio}
q(\calR)=\frac{d_{\textrm{\euro}\textrm{\yen}}(\calR)}{d_{\textrm{\euro}\textrm{\pounds}}(\calR)}.
\end{equation}
\begin{theorem}\label{irratBC}
Let the number \eqref{ratio} be irrational. Then set
$S^{bal}(\calR)$ is dense in the totality $T^{bal}$ of all possible
balanced ensembles.
\end{theorem}
A proof of this assertion will be given later on.

Consider also the case where $q=q(\calR)$ is a rational number:
$q=m/n$ with co-prime integers $m,n$ (including the possibilities
$m=1$ or $n=1$). Denote also
\[
\label{aplhaB} \alpha=\exp (d_{\textrm{\euro}\textrm{\yen}}/n).
\]

The following assertion is a straightforward analog of Proposition
\ref{rep1P}.
\begin{proposition}\label{repBP}
The inclusions $ S(\calR)\subset T_{\alpha}(\calR) $ and
$S^{bal}(\calR)\subset T^{bal}_{\alpha}(\calR)$ hold.
\end{proposition}

The following is an analog of Theorem~\ref{arbH}:

\begin{proposition}\label{arbBH}
The equality $S^{bal}(\calR)= T^{bal}(\calR)$ holds.
\end{proposition}

A proof of this assertion will be given later on.

Note that the expressions like (\ref{estN}) are not valid in
general. Similar expressions may be established, however, for the
cases $m=1$ or $n=1$.  Note also that the case when the discrepancy
triplet is of one the forms $(a,a,0)$ or $(a,0,-a)$ or $(0,a,a)$,
$a\not=0$, was implicitly considered in Section~\ref{directSS}: see
the first line in \eqref{fori} and Proposition \ref{arbA2H}.

\subsection{Case C\label{caseC}}
Consider the case where all three arbitrage discrepancies
\eqref{discrep} are not equal to zero.

\begin{corollary}\label{irratCC}
Let at least one of the ratios
\begin{equation}\label{ratiose}
q_{1}(\calR)=
\frac{d_{\textrm{\euro}\textrm{\yen}}(\calR)}{d_{\textrm{\euro}\textrm{\pounds}}(\calR)},
\quad
q_{2}(\calR)=\frac{d_{\textrm{\pounds}\textrm{\yen}}(\calR)}{d_{\textrm{\euro}\textrm{\pounds}}(\calR)}
\end{equation}
be irrational. Then the set $S^{bal}(\calR)$ is dense in the
totality $T^{bal}$ of all possible balanced ensembles.
\end{corollary}

Suppose now that both ratios \eqref{ratiose} are rational:
\[
q_{1}(\calR)=\frac{m_{1}}{n_{1}},\quad
q_{2}(\calR)=\frac{m_{2}}{n_{2}}.
\]
Denote by $\lcm(n_1,n_2)$ the least common multiple of the
corresponding denominators. Denote
\[
\alpha(\calR)=\exp\left(\frac{d_{\textrm{\euro}\textrm{\pounds}}(\calR)}{\lcm(n_1,n_2)}\right).
\]
\begin{proposition}\label{rep1CP}
The relationships $ S(\calR)\subset T_{\alpha}(\calR)$ and
$S^{bal}(\calR)\subset T^{bal}_{\alpha}(\calR)$ hold.
\end{proposition}

\begin{corollary}\label{arbCH}
Let
\begin{equation}
\label{resC} \lcm(n_1,n_2)=n_{1}\cdot n_{2}.
\end{equation}
Then $S^{bal}(\calR)= T^{bal}_{\alpha}(\calR)$.
\end{corollary}
\begin{proof}
This assertion as well as formulated below Corollary~\ref{arbC2H}
follows from Proposition~\ref{arbBH} together with Lemma
\ref{colP}.
\end{proof}

Consider finally the case when the ratios $q_{1}(\calR)$ and
$q_{2}(\calR)$ are rational, but \eqref{resC} does not hold. In
this case we introduce the number $\gamma$ such that
$d_{i}=k_{i}\gamma$ where the numbers $k_{i}$ are integers and
their greatest common divisor, $\gcd(k_{1},k_{2},k_{3})$, is equal
to $1$. Consider also the following six numbers:
\begin{equation}\label{lcd2}
\begin{alignedat}{2}
a_{1}&=\gcd(k_{1},k_{2}), &a_{2}&=\gcd(k_{1},k_{3}),\\[1mm]
a_{3}&=\gcd(k_{2},k_{3}), &a_{4}&=\gcd(k_{1},k_{2}-k_{3}),\\[1mm]
a_{5}&=\gcd(k_{2},k_{1}+k_{3}),\quad &a_{6}&=\gcd(k_{3},k_{1}-k_{2}).
\end{alignedat}
\end{equation}
Introduce also the numbers $\alpha_{i}=\exp a_{i}$, $i=1,\ldots,
6$.

\begin{corollary}\label{arbC2H}
The equation $S^{bal}(\calR)=\cup_{i=1}^{6}
T^{bal}_{\alpha_{i}}(\calR)$ holds.
\end{corollary}
Note that all six numbers in \eqref{lcd2}  may indeed be greater
than one. For instance, consider: $k_1 = 595$, $k_2 = 1683$, $k_3 =
308$. By inspection, $\gcd(k_1, k_2, k_3)=1$, and
\begin{alignat*}{3}
 a_1&=\gcd(k_1, k_2)=17,& a_2&=\gcd(k_1, k_3)=7,\\
 a_3&=\gcd(k_2, k_3)=11,& a_4&=\gcd(k_1 - k_2, k_3)=4,\\
  a_5&=\gcd(k_1 + k_3, k_2)=3,\quad& a_6&=\gcd(k_1, k_2 - k_3)=5.
\end{alignat*}

\section{Proofs}\label{S-proofs}

From this point onward we discuss the proofs of the theorems
formulated above. This part of the paper is organised as follows.
In Section~\ref{S-sarb} we introduce, as a useful auxiliary tool,
stronger arbitrage procedures. Using strong arbitrages, we
``linearise the problem'', reducing it to investigation of all
possible products of 12 explicitly written $6\times 6$-matrices.
Afterwards, in Section~\ref{S-scs} we separate a family of 12
$3\times 3$-matrices $G^{(i)}$ such that the products of these
matrices completely describe the dynamics of the discrepancy
triplets. The properties of such products appear to be of key
importance, and these are investigated in Section~\ref{S-struct}.
The results are applied  in Section~\ref{S-discrep}.
Sections~\ref{S-incdyn} and \ref{S-proofT} are dedicated to
finalising the proof of Theorem~\ref{arbH}. Finally, in
Sections~\ref{knotsSS}--\ref{fp2SS} we provide proofs for
Theorem~\ref{irratBC} and Proposition \ref{arbBH}.

\subsection{Strong Arbitrages}\label{S-sarb}

We use, as an auxiliary tool, stronger arbitrage procedures. Let us
begin with an example. Consider the currencies triplet
$(\textrm{\$}\textrm{\euro}\textrm{\pounds})$. For a given $\calR$
we define the strong arbitrage
$\Hat\calA_{\textrm{\$}\textrm{\euro}\textrm{\pounds}}\calR$ as
$\calA_{\textrm{\$}\textrm{\euro}\textrm{\pounds}}$ if the
inequality \eqref{unb} holds, and as
$\calA_{\textrm{\euro}\textrm{\$}\textrm{\pounds}}$, otherwise.
Note that in both cases the result in terms of principal exchange
rates is the same: the rate $r_{\textrm{\$}\textrm{\euro}}$ is
changed to $r_{\textrm{\$}\textrm{\euro}}^{new}=
\frac{r_{\textrm{\$}\textrm{\yen}}}{r_{\textrm{\euro}\textrm{\pounds}}}$.

The strong arbitrage
$\Hat\calA_{\textrm{\$}\textrm{\euro}\textrm{\yen}}$ is the second
entry in Table~\ref{starbT} of the possible 12 strong arbitrages.
The meaning of a strong arbitrage is simple. This is an arbitrage
balancing a sub-FX market such as
$\textrm{\$}\textrm{\euro}\textrm{\yen}$ by changing the exchange
rate for a pair such as $Dollar \leftrightarrows Euro$.  We will
use, where convenient, the notation $\Hat\calA^{(n)}$ for the
arbitrage number $n$ from this table.

\begin{table}[!htbp]
\caption{Strong arbitrages}\label{starbT}
\begin{tabular}{llll}
Number& Strong arbitrage& Action& Numbers of arbitrages\\
\hline \\

1& $\Hat\calA_{\textrm{\$}\textrm{\euro}\textrm{\pounds}}$& $
r^{new}_{\textrm{\$}\textrm{\euro}}=
r_{\textrm{\$}\textrm{\pounds}}\cdot
r_{\textrm{\euro}\textrm{\pounds}}^{-1}$ & 1, 7
\\

2& $\Hat\calA_{\textrm{\$}\textrm{\euro}\textrm{\yen}}$  & $
r^{new}_{\textrm{\$}\textrm{\euro}}=
r_{\textrm{\$}\textrm{\yen}}\cdot
r_{\textrm{\euro}\textrm{\yen}}^{-1} $ &2, 8
\\

3& $\Hat\calA_{\textrm{\$}\textrm{\pounds}\textrm{\euro}}$  & $
r^{new}_{\textrm{\$}\textrm{\pounds}}=
r_{\textrm{\$}\textrm{\euro}}\cdot
r_{\textrm{\euro}\textrm{\pounds}}$ &3, 13
\\

4& $\Hat\calA_{\textrm{\$}\textrm{\pounds}\textrm{\yen}}$& $
r^{new}_{\textrm{\$}\textrm{\pounds}}=
r_{\textrm{\$}\textrm{\yen}}\cdot
r_{\textrm{\pounds}\textrm{\yen}}^{-1}$ & 4, 14
\\

5& $\Hat\calA_{\textrm{\$}\textrm{\yen}\textrm{\euro}}$ & $
r^{new}_{\textrm{\$}\textrm{\yen}}= r_{\textrm{\$},
\textrm{\euro}}\cdot r_{\textrm{\euro}\textrm{\yen}}$ &5, 19
\\

6 &$\Hat\calA_{\textrm{\$}\textrm{\yen}\textrm{\pounds}}$&
 $ r^{new}_{\textrm{\$}\textrm{\yen}}= r_{\textrm{\$}\textrm{\pounds}}\cdot r_{\textrm{\pounds}\textrm{\yen}}$
& 6, 20
\\

7& $\Hat\calA_{\textrm{\euro}\textrm{\pounds}\textrm{\$}}$&
  $ r^{new}_{\textrm{\euro}\textrm{\pounds}}= r_{\textrm{\$}\textrm{\pounds}}\cdot r_{\textrm{\$}\textrm{\euro}}^{-1}$
 & 9, 15
\\

8& $\Hat\calA_{\textrm{\euro}\textrm{\pounds}\textrm{\yen}}$&
 $ r^{new}_{\textrm{\euro}\textrm{\pounds}}= r_{\textrm{\euro}\textrm{\yen}}\cdot r_{\textrm{\pounds}\textrm{\yen}}^{-1}$
 & 10, 16
\\

9& $\Hat\calA_{\textrm{\euro}\textrm{\yen}\textrm{\$}}$&
 $ r^{new}_{\textrm{\euro}\textrm{\yen}}= r_{\textrm{\$}\textrm{\yen}}\cdot r_{\textrm{\$}\textrm{\euro}}^{-1}$
  & 11, 21
\\

10& $\Hat\calA_{\textrm{\euro}\textrm{\yen}\textrm{\pounds}}$&
 $ r^{new}_{\textrm{\euro}\textrm{\yen}}= r_{\textrm{\euro}\textrm{\pounds}}\cdot r_{\textrm{\pounds}\textrm{\yen}}$
 & 12, 22
\\

11& $\Hat\calA_{\textrm{\pounds}\textrm{\yen}\textrm{\$}}$&
 $ r^{new}_{\textrm{\pounds}\textrm{\yen}}= r_{\textrm{\$}\textrm{\yen}}\cdot r_{\textrm{\$}\textrm{\pounds}}^{-1}$
 & 17, 23
\\

12& $\Hat\calA_{\textrm{\pounds}\textrm{\yen}\textrm{\euro}}$& $
r^{new}_{\textrm{\pounds}\textrm{\yen}}=
r_{\textrm{\euro}\textrm{\yen}}\cdot
r_{\textrm{\euro}\textrm{\pounds}}^{-1}$ & 18, 24
\end{tabular}
\end{table}

\begin{proposition}
For any  arbitrage chain \eqref{arbsec}, and any initial exchange
rates $\calR$, there exists a chain $\hbA=\Hat\calA_1 \cdots
\Hat\calA_n $ of strong arbitrages such that
$\calR\hbA=\calR{\bA}$. Conversely, for any chain $\hbA=\Hat\calA_1
\cdots \Hat\calA_n $ of strong arbitrages, and any initial exchange
rates $\calR$, there exists a chain of arbitrages such that
$\calR\hbA=\calR{\bA}$.
\end{proposition}

This proposition reduces investigation of the questions from the
previous section to investigation of analogous questions related to
chains of strong arbitrages.

Now we relate each strong arbitrage to a $6\times 6$ matrix
$B(\calA)$ as follows:
{\small\[ B_{\textrm{\$}\textrm{\euro}
\textrm{\pounds}}=B^{(1)} = \left(
\begin{smallmatrix-mod}
      \zer& \zer& \zer& \zer& \zer& \zer \\ -1& 1& \zer& \zer& \zer& \zer  \\ \zer& \zer& 1& \zer& \zer& \zer\\
 1& \zer& \zer& 1& \zer& \zer \\ \zer& \zer& \zer& \zer& 1& \zer \\ \zer& \zer& \zer& \zer& \zer& 1
\end{smallmatrix-mod}
\right),\quad
 B_{\textrm{\$} \textrm{\euro}\textrm{\yen}}=B^{(2)}=
\left(
\begin{smallmatrix-mod}
     \zer& \zer& \zer& \zer& \zer& \zer \\ \zer& 1& \zer& \zer& \zer& \zer \\ \zer& \zer& 1& \zer& \zer& \zer \\
 \zer& \zer& \zer& 1& \zer& \zer \\ -1& \zer& \zer& \zer& 1& \zer \\ 1& \zer& \zer& \zer& \zer& 1
\end{smallmatrix-mod}
\right),
\]
\[
B_{\textrm{\$} \textrm{\pounds} \textrm{\euro}}=B^{(3)}=
\left(
\begin{smallmatrix-mod}
      1& \zer& \zer& 1& \zer& \zer \\ \zer& 1& \zer& 1& \zer& \zer \\ \zer& \zer& 1& \zer& \zer& \zer\\
  \zer& \zer& \zer& \zer& \zer& \zer \\ \zer& \zer& \zer& \zer& 1& \zer \\ \zer& \zer& \zer& \zer& \zer& 1
\end{smallmatrix-mod}
\right),\quad
B_{\textrm{\$} \textrm{\pounds} \textrm{\yen}}=B^{(4)}=
\left(
\begin{smallmatrix-mod}
 1& \zer& \zer& \zer& \zer& \zer \\ \zer& 1& \zer& \zer& \zer& \zer \\ \zer& \zer& 1& -1& \zer& \zer\\
  \zer& \zer& \zer& \zer& \zer& \zer \\ \zer& \zer& \zer& \zer& 1& \zer \\ \zer& \zer& \zer& 1& \zer& 1
\end{smallmatrix-mod}
\right),
\]
\[
B_{\textrm{\$} \textrm{\yen} \textrm{\euro}}=B^{(5)}=
\left
(\begin{smallmatrix-mod}
1& \zer& \zer& \zer& \zer& 1\\ \zer& 1& \zer& \zer& \zer& \zer \\ \zer& \zer& 1& \zer& \zer& \zer\\
  \zer& \zer& \zer& 1& \zer& \zer \\ \zer& \zer& \zer& \zer& 1& 1\\ \zer& \zer& \zer& \zer& \zer& \zer
\end{smallmatrix-mod}
\right),\quad
B_{\textrm{\$} \textrm{\yen} \textrm{\pounds}}=B^{(6)}=
\left
(\begin{smallmatrix-mod}
 1& \zer& \zer& \zer& \zer& \zer \\ \zer& 1& \zer& \zer& \zer& \zer \\ \zer& \zer& 1& \zer& \zer& 1\\
  \zer& \zer& \zer&  1& \zer& 1\\ \zer& \zer& \zer& \zer& 1& \zer \\ \zer& \zer& \zer& \zer& \zer& \zer
\end{smallmatrix-mod}
\right),
\]
\[
B_{\textrm{\euro} \textrm{\pounds} \textrm{\$}}=B^{(7)}=
\left(\begin{smallmatrix-mod}
1& -1& \zer& \zer& \zer& \zer \\ \zer& \zer& \zer& \zer& \zer& \zer \\ \zer& \zer& 1& \zer& \zer& \zer \\
   \zer& 1& \zer& 1& \zer& \zer \\ \zer& \zer& \zer& \zer& 1& \zer \\ \zer& \zer& \zer& \zer& \zer& 1
\end{smallmatrix-mod}\right),\quad
B_{\textrm{\euro} \textrm{\pounds} \textrm{\yen}}=B^{(8)}=
\left(\begin{smallmatrix-mod}
 1& \zer& \zer& \zer& \zer& \zer \\ \zer& \zer& \zer& \zer& \zer& \zer \\ \zer& -1& 1& \zer& \zer& \zer \\
   \zer& \zer& \zer& 1& \zer& \zer \\ \zer& 1& \zer& \zer& 1& \zer \\ \zer& \zer& \zer& \zer& \zer& 1
\end{smallmatrix-mod}\right),
\]
\[
B_{\textrm{\euro} \textrm{\yen} \textrm{\$}}=B^{(9)}=
\left(\begin{smallmatrix-mod}
1& \zer& \zer& \zer& -1& \zer \\ \zer& 1& \zer& \zer& \zer& \zer \\ \zer& \zer& 1& \zer& \zer& \zer \\
   \zer& \zer& \zer& 1& \zer& \zer \\ \zer& \zer& \zer& \zer& \zer& \zer \\ \zer& \zer& \zer& \zer& 1& 1
\end{smallmatrix-mod}\right),\quad
B_{\textrm{\euro} \textrm{\yen} \textrm{\pounds}}=B^{(10)}=
\left(\begin{smallmatrix-mod}
 1& \zer& \zer& \zer& \zer& \zer \\ \zer& 1& \zer& \zer& 1& \zer \\ \zer& \zer& 1& \zer& 1& \zer \\
    \zer& \zer& \zer& 1& \zer& \zer \\ \zer& \zer& \zer& \zer& \zer& \zer \\ \zer& \zer& \zer& \zer& \zer& 1
\end{smallmatrix-mod}\right),
\]
\[
B_{\textrm{\pounds} \textrm{\yen} \textrm{\$}}=B^{(11)}=
\left(\begin{smallmatrix-mod}
1& \zer& \zer& \zer& \zer& \zer \\ \zer& 1& \zer& \zer& \zer& \zer \\ \zer& \zer& \zer& \zer& \zer& \zer \\
   \zer& \zer& -1& 1& \zer& \zer \\ \zer& \zer& \zer& \zer& 1& \zer \\ \zer& \zer& 1& \zer& \zer& 1
\end{smallmatrix-mod}\right),\quad
B_{\textrm{\pounds} \textrm{\yen} \textrm{\euro}}=B^{(12)}=
\left(\begin{smallmatrix-mod}
1& \zer& \zer& \zer& \zer& \zer \\ \zer& 1& -1& \zer& \zer& \zer \\ \zer& \zer& \zer& \zer& \zer& \zer \\
   \zer& \zer& \zer& 1& \zer& \zer \\ \zer& \zer& 1& \zer& 1& \zer \\ \zer& \zer& \zer& \zer& \zer& 1
\end{smallmatrix-mod}\right).
\]}

For any ensemble $\calR = \left(r_{\textrm{\$}\textrm{\euro}},
r_{\textrm{\$}\textrm{\pounds}},r_{\textrm{\$}\textrm{\yen}},
r_{\textrm{\euro}\textrm{\pounds}},
r_{\textrm{\euro}\textrm{\yen}}, r_{\textrm{\pounds}\textrm{\yen}}
\right)$ we denote
\[
\log\calR =\left(\log r_{\textrm{\$}\textrm{\euro}},\ \log
r_{\textrm{\$}\textrm{\pounds}},\ \log
r_{\textrm{\$}\textrm{\yen}},\
 \log r_{\textrm{\euro}\textrm{\pounds}},\ \log r_{\textrm{\euro}\textrm{\yen}},\
 \log r_{\textrm{\pounds}\textrm{\yen}}
   \right).
  \]

\begin{proposition}\label{oldprop}
The equation $\log (\calR \Hat\calA^{(i)}) = (\log \calR) B^{(i)}$
holds for $i=1,\ldots , 12$.
\end{proposition}
\begin{proof}
Follows from definitions.
\end{proof}

\subsection{A Special Coordinate System}\label{S-scs}
In the six-dimensional real coordinate space $\bbR^{6}$ we
introduce the vectors
\[
{\bv}_{1}=(1,-1,0,1,0,0), \
{\bv}_{2}=(1,0,-1,0,1,0), \
{\bv}_{3}=(0,1,-1,0,0,1).
\]
By definition for any ensemble $\calR$
\[
\langle {\bv}_{1}, \log \calR \rangle= d_{\textrm{\euro}\textrm{\pounds}}(\calR) , \
\langle {\bv}_{2}, \log \calR \rangle= d_{\textrm{\euro}\textrm{\yen}}(\calR), \
\langle {\bv}_{3}, \log \calR \rangle= d_{\textrm{\pounds}\textrm{\yen}}(\calR),
\]
where $\langle \cdot, \cdot \rangle$ denotes the usual inner
product in $\bbR^{6}$.

Propositions~\ref{balp} and~\ref{oldprop} together imply
\begin{corollary}
The three-dimensional subspace $\langle {\bv}_{1}, {\bv} \rangle =
\langle {\bv}_{2}, {\bv} \rangle
 =\langle {\bv}_{3}, {\bv} \rangle =0$
is invariant with respect to each linear operator ${\bv}\to
{\bv}B^{(i)}$, $i=1,\ldots, 12$.
\end{corollary}

We introduce in $\bbR^{6}$ the new basis
\[
\{{\be}_{1},{\be}_{2},{\be}_{3},{\bv}_{1},{\bv}_{2},{\bv}_{3}\};
\]
here ${\be}_{1}=(1,0,0,0,0,0)$, ${\be}_{2}=(0,1,0,0,0,0)$,
${\bv}_{3}=(0,0,1,0,0,0)$. By the last corollary  in this basis the
matrices of the linear operators ${\bv}\to {\bv}B^{(i)}$ have the
block-triangular form:
\[
D^{(i)}=\left(
\begin{array}{ll}
{\bf 1} & {\bf 0} \\
H^{(i)} & G^{(i)}
\end{array}
\right) .
\]
Here
\[
{\bf 0}=\left(\begin{array}{lll}
0    & 0  & 0\\
0 & 0  & 0\\
0 & 0  & 0
\end{array}\right),
 \quad
 {\bf 1}=\left(\begin{array}{lll}
1    & 0  & 0\\
0 & 1  & 0\\
0 & 0  & 1
\end{array}\right),
\]
and $G^{(i)}, H^{(i)}$ are some $3\times 3$-matrices.

Denote
\[
 {\bv}(\calR)= \left(
 \log r_{\textrm{\$}\textrm{\euro}},\
 \log r_{\textrm{\$}\textrm{\pounds}},\
 \log r_{\textrm{\$}\textrm{\yen}},\
 d_{\textrm{\euro}\textrm{\pounds}}(\calR),\
 d_{\textrm{\euro}\textrm{\yen}}(\calR),\
 d_{\textrm{\pounds}\textrm{\yen}}(\calR)
   \right).
\]

\begin{proposition}\label{blockC}
The equality $ {\bv}(\calR \Hat{\calA}^{(i)})={\bv}(\calR)D^{(i)} $
holds for $i=1,\ldots, 12$.
\end{proposition}

\begin{proof}
Follows from Lemma~\ref{Qlem} and Proposition~\ref{oldprop}.
\end{proof}

The matrices  $D^{(i)} $ may be written explicitly as
\begin{equation}\label{expl}
QB^{(i)}Q^{-1},
\end{equation}
where
\begin{equation}\label{explQ}
Q=\left(\begin{smallmatrix-mod}
1& \zer& \zer& \zer & \zer & \zer \\
 0& 1& \zer&1& \zer& \zer \\
 0& \zer& 1& \zer& \zer & \zer \\
   1& -1& \zer& 1& \zer& \zer \\
  1& \zer& -1& \zer& 1& \zer \\
   0& 1& -1& \zer& \zer& 1
\end{smallmatrix-mod}\right),\quad
Q^{-1}=\left(\begin{smallmatrix-mod}
1& \zer& \zer& -1& -1& \zer \\
 0& 1& \zer&1& \zer& -1\\
 0& \zer& 1& \zer& 1& 1\\
   0& \zer& \zer& 1& \zer& \zer \\
  0& \zer& \zer& \zer& 1& \zer \\
   0& \zer& \zer& \zer& \zer& 1
\end{smallmatrix-mod}\right).
\end{equation}
\begin{lemma}\label{Qlem}
The following equations are valid:
\[
G^{(1)}\hphantom{^{0}}=\left(\begin{smallmatrix-mod}
\zer  & -1    & \zer \\
\zer  & 1     & \zer \\
\zer  & \zer   & 1
\end{smallmatrix-mod}\right),\quad
G^{(2)}\hphantom{^{0}}=\left(\begin{smallmatrix-mod}
1    & \zer   & \zer \\
-1   & \zer   & \zer \\
\zer  & \zer   & 1
\end{smallmatrix-mod}\right),\quad
G^{(3)}\hphantom{^{0}}=\left(\begin{smallmatrix-mod}
\zer  & \zer   & 1   \\
\zer  & 1     & \zer \\
\zer  & \zer   & 1
\end{smallmatrix-mod}\right),
\]
\[
G^{(4)}\hphantom{^{0}}=\left(\begin{smallmatrix-mod}
1    & \zer   & \zer \\
\zer  & 1     & \zer \\
1    & \zer   & \zer
\end{smallmatrix-mod}\right),\quad
G^{(5)}\hphantom{^{0}}=\left(\begin{smallmatrix-mod}
1    & \zer   & \zer \\
\zer  & \zer   & -1  \\
\zer  & \zer   & 1
\end{smallmatrix-mod}\right),\quad
G^{(6)}\hphantom{^{0}}=\left(\begin{smallmatrix-mod}
1    & \zer   & \zer \\
\zer  & 1     & \zer \\
\zer  & -1    & \zer
\end{smallmatrix-mod}\right),
\]
\[
G^{(7)}\hphantom{^{0}}=\left(\begin{smallmatrix-mod}
\zer  & \zer   & \zer \\
\zer  & 1     & \zer \\
\zer  & \zer   & 1
\end{smallmatrix-mod}\right),\quad
G^{(8)}\hphantom{^{0}}=\left(\begin{smallmatrix-mod}
\zer  & \zer   & \zer \\
1    & 1     & \zer \\
-1   & \zer   & 1
\end{smallmatrix-mod}\right),\quad
G^{(9)}\hphantom{^{0}}=\left(\begin{smallmatrix-mod}
1    & \zer   & \zer \\
\zer  & \zer   & \zer \\
\zer  & \zer   & 1
\end{smallmatrix-mod}\right),
\]
\[
G^{(10)}=\left(\begin{smallmatrix-mod}
1    & 1     & \zer \\
\zer  & \zer   & \zer \\
\zer  & 1     & 1
\end{smallmatrix-mod}\right),\quad
G^{(11)}=\left(\begin{smallmatrix-mod}
1    & \zer   & \zer \\
\zer  & 1     & \zer \\
\zer  & \zer   & \zer
\end{smallmatrix-mod}\right),\quad
G^{(12)}=\left(\begin{smallmatrix-mod}
1    & \zer   & -1  \\
\zer  & 1     & 1   \\
\zer  & \zer   & \zer
\end{smallmatrix-mod}\right),
\]
and
\[
H^{(1)}=\left( \begin{smallmatrix-mod}
 -1&\zer &\zer \\
  \zer &\zer &\zer \\
\zer &\zer &\zer
 \end{smallmatrix-mod}\right),\quad
H^{(2)}=\left( \begin{smallmatrix-mod}
 \zer &\zer &\zer \\
  -1&\zer &\zer \\
\zer &\zer &\zer
\end{smallmatrix-mod}\right),\quad
H^{(3)}=\left( \begin{smallmatrix-mod}
 \zer &1&\zer \\
 \zer &\zer &\zer \\
\zer &\zer &\zer
\end{smallmatrix-mod}\right),
\]
\[
H^{(4)}=\left( \begin{smallmatrix-mod}
 \zer &\zer &\zer \\
 \zer &\zer &\zer \\
\zer &-1&\zer
\end{smallmatrix-mod}\right),\quad
H^{(5)}=\left( \begin{smallmatrix-mod}
 \zer &\zer &\zer \\
  \zer &\zer &1\\
\zer &\zer &\zer
 \end{smallmatrix-mod}\right),\quad
H^{(6)}=\left( \begin{smallmatrix-mod}
 \zer &\zer &\zer \\
  \zer &\zer &\zer \\
\zer &\zer &1
\end{smallmatrix-mod}\right),
\]
\[
H^{(i)}= { \bf 0},\quad i=7, \ldots , 12 .
\]
\end{lemma}
\begin{proof}
Follows by inspection from \eqref{expl}, \eqref{explQ}.
\end{proof}

\begin{proposition}\label{disc1P}
The discrepancy ensemble $\calD(\calR \Hat\calA)$ depends only on
$\calD(\calR)$ and $\Hat\calA$, and may be written as follows: $
\calD(\calR \Hat\calA^{(i)})= \calD(\calR) G^{(i)}$.  Here $i$ is
the number of a strong arbitrage as listed in Table~\ref{starbT}.
\end{proposition}

\begin{proof}
Follows from Proposition~\ref{blockC}.
\end{proof}

By the last proposition a discrepancy ensemble $\calD(\calR \hbA)$
related to an arbitrage chain $\hbA=\Hat\calA_{1}\cdots
\Hat\calA_{n}$ may be written as
\[
\calD(\calR \hbA)=\calD(\calR)\prod_{i=1}^{n}G_{i}.
\]
Therefore the set ${\bbG}$ of all possible products of the matrices
$G^{(i)}$ is of interest.

\subsection{Structure of the Set ${\bbG}$}\label{S-struct}

The following assertion is the key observation of our paper:
\begin{lemma}\label{Gfini}
The set  ${\bbG}$ consists of 229 elements.
\end{lemma}

\begin{proof}
By inspection
\end{proof}

Denote by $\Hat{\mathbb A}$ the totality of all finite chains of
strong arbitrages.

\begin{corollary}\label{Dfini}
For a given $\calR$ the set $ \bbD(\calR)=\{\calD(\calR \hbA): \
{\bA}\in {\mathbb A}\} $ consists of less than 230 elements.
\end{corollary}

Let us discuss briefly the structure of the set ${\bbG}$.  A subset
${\bG}$ of ${\bbG}$ is called a connected component, if for any
$G_{1},G_{2}\in {\bG}$ there exists $G\in {\bbG}$ satisfying
$G_{2}=G_{1}G$.  By the definition different connected components
do not intersect.
\begin{lemma}\label{comp}
The set ${\bbG}$ is partitioned into 14 connected components
$U_{1},\ldots U_{14}$. Each of the first six connected components
includes 24 matrices of range 2; each of the connected components
$U_{7}, \ldots U_{13}$ includes 12 matrices of range one; the last
component contains a single zero matrix.
\end{lemma}

The sets $U_{1},\ldots, U_{6}$ may be characterised by the following
inclusions:
\[
G^{(2i-1)},\ G^{(2i)}\in U_{i},\quad i=1, \ldots ,6.
\]
To identify the connected components $U_{7}, \ldots U_{13}$ we list
below the smallest lexicographical matrices from these components
\[
\left( \begin{smallmatrix-mod}
 -1&-1&\zer \\
  \zer &\zer &\zer \\
\zer &\zer &\zer
 \end{smallmatrix-mod}\right)\in U_{7},\hphantom{_{0}}\quad
\left( \begin{smallmatrix-mod}
 \zer &\zer &\zer \\
  -1&-1&\zer \\
\zer &\zer &\zer
\end{smallmatrix-mod}\right)\in U_{8},\hphantom{_{0}}\quad
\left( \begin{smallmatrix-mod}
 \zer &\zer &\zer \\
 \zer &\zer &\zer \\
-1&-1&\zer
\end{smallmatrix-mod}\right)\in U_{9},\hphantom{_{0}}
\]
\[
\left( \begin{smallmatrix-mod}
 -1&-1&\zer \\
 1&1&\zer \\
\zer &\zer &\zer
\end{smallmatrix-mod}\right)\in U_{10},\quad
\left( \begin{smallmatrix-mod}
 -1&-1&\zer \\
  \zer &\zer &\zer \\
-1&-1&\zer
 \end{smallmatrix-mod}\right)\in U_{11},\quad
\left( \begin{smallmatrix-mod}
 \zer &\zer &\zer \\
  -1&-1&\zer \\
1&1&\zer
\end{smallmatrix-mod}\right)\in U_{12},
\]
\[
\left( \begin{smallmatrix-mod}
 -1&-1&\zer \\
  1&1&\zer \\
-1&-1&\zer
\end{smallmatrix-mod}\right)\in U_{13}.
\]

One can move from one connected component $U_{i}$ to another
component $U_{j}$ applying a matrix $G^{(i)}$, $i=1,\ldots, 12$.
Let us describe the set of possible transitions. We will use the
notation $U_{i} \succ U_{j}$ if such a transition is possible.
\begin{lemma}\label{trans}
The following relationships hold:
\begin{alignat*}{3}
U_{1} &\succ U_{9},U_{10},U_{13},\quad & U_{2} &\succ
U_{8},U_{11},U_{13},\quad &U_{3} &\succ
U_{7},U_{12},U_{13}, \\
U_{4} &\succ U_{8},U_{9},U_{12},\quad & U_{5} &\succ
U_{7},U_{9},U_{11},\quad &U_{6} &\succ U_{7},U_{8},U_{10}.
\end{alignat*}
Also $U_{i} \succ U_{14}$, $i=1, \ldots,13$.
\end{lemma}

\begin{proof} By inspection.
\end{proof}

\begin{lemma}
For any $G\in{\bbG}$ either $G$ or $G^2$ or $G^3$ is a projector.
\end{lemma}

\begin{proof} By inspection.
\end{proof}

\subsection{Discrepancy Dynamics}\label{S-discrep}
The structure of the set ${\bbG}$ explained above induces
structuring of the set of discrepancies, which we discuss below. We
say that a set ${\bD}$ of discrepancies is a connected component if
for any $\calD_1,\calD_2\in {\bD}$ there exists an arbitrage chain
${\bA}$ satisfying $ \calD_1{\bA}=\calD_2$.  For a given reals
$a,b$ we denote by ${\bD}(a,b)$ the set of different triplets from
the collection
\begin{equation}\label{24disc}
\begin{alignedat}{2}
\calD_{1}(a,b)&=\left(a,b,-a+b\right),\quad& \calD_{2}(a,b)&=\left(-a+b,b,a\right),\\
\calD_{3}(a,b)&=\left(a,a-b,-b\right),\quad& \calD_{4}(a,b)&=\left(-a+b,-a,-b\right),\\
\calD_{5}(a,b)&=\left(-b,a-b,a\right),\quad& \calD_{6}(a,b)&=\left(-b,-a,-a+b\right),\\
\calD_{7}(a,b)&=\left(0,b,-a+b\right),\quad&\calD_{8}(a,b)&=\left(a,0,-a+b\right),\\
\calD_{9}(a,b)&=\left(a,b,0\right),\quad& \calD_{10}(a,b)&=\left(0,b,a\right),\\
\calD_{11}(a,b)&=\left(-a+b,0,a\right),\quad& \calD_{12}(a,b)&=\left(-a+b,b,0\right),\\
\calD_{13}(a,b)&=\left(0,-a,-b\right),\quad& \calD_{14}(a,b)&=\left(-a+b,0,-b\right),\\
\calD_{15}(a,b)&=\left(-a+b,-a,0\right),\quad& \calD_{16}(a,b)&=\left(0,a-b,-b\right),\\
\calD_{17}(a,b)&=\left(a,0,-b\right),\quad& \calD_{18}(a,b)&=\left(a,a-b,0\right),\\
\calD_{19}(a,b)&=\left(0,a-b,a\right),\quad& \calD_{20}(a,b)&=\left(-b,0,a\right),\\
\calD_{21}(a,b)&=\left(-b,a-b,0\right),\quad& \calD_{22}(a,b)&=\left(0,-a,-a+b\right),\\
\calD_{23}(a,b)&=\left(-b,0,-a+b\right),\quad& \calD_{24}(a,b)&=\left(-b,-a,0\right).
\end{alignedat}
\end{equation}

\begin{lemma}
Each set ${\bD}(a,b)$ is a connected component, and each connected
component coincides with a certain set ${\bD}(a,b)$.
\end{lemma}

\begin{proof}
This statement may be proved by inspection.
\end{proof}

Let us discuss in brief the structure of the sets ${\bD}(a,b)$ for
different values $a,b$.  Clearly, ${\bD}(0,0)$ consists of the
single zero triplet $\calD_{0}=(0,0,0)$.  The connected components
${\bD}(\pm a,0)$, ${\bD}(0,\pm a)$, ${\bD}(a,a)$, ${\bD}(-a,-a)$
coincide and include the following 12 elements:
\begin{equation}\label{12disc}
\begin{alignedat}{2}
\calD_{1}(a)&=a(\hphantom{-}0,\hphantom{-}0,\hphantom{-}1),\quad&
\calD_{2}(a)&=a(-1,\hphantom{-}0,\hphantom{-}1),\\
\calD_{3}(a)&=a(-1,\hphantom{-}0,\hphantom{-}0),\quad&
\calD_{4}(a)&=a(-1,-1,\hphantom{-}0),\\
\calD_{5}(a)&=a(\hphantom{-}0,-1,\hphantom{-}0),\quad&
\calD_{6}(a)&=a(\hphantom{-}0,-1,-1),\\
\calD_{7}(a)&=a(\hphantom{-}0,\hphantom{-}0,-1),\quad&
\calD_{8}(a)&=a(\hphantom{-}1,\hphantom{-}0,-1),\\
\calD_{9}(a)&=a(\hphantom{-}1,\hphantom{-}0,\hphantom{-}0),\quad&
\calD_{10}(a)&=a(\hphantom{-}1,\hphantom{-}1,\hphantom{-}0),\\
\calD_{11}(a)&=a(\hphantom{-}0,\hphantom{-}1,\hphantom{-}0),\quad&
\calD_{12}(a)&=a(\hphantom{-}0,\hphantom{-}1,\hphantom{-}1).
\end{alignedat}
\end{equation}
We use notation ${\bD}(a)$ for this set. Geometrically the set
${\bD}(a)$ represents vertices of a partly distorted truncated
cuboctahedron, or triangular orthobicupola, shown in Fig.~\ref{tc}.
The structure of this component will be explained in more detail in
Section~\ref{S-proofT}. The set ${\bD}(a, -a)$, ${\bD}(a, 2a)$,
${\bD}(a, a/2)$, also consists of 12 elements. Geometrically these
sets ${\bD}(a)$ represent vertices of a distorted truncated
tetrahedron, shown in Fig.~\ref{tt}. Otherwise, a set ${\bD}(a,
b)$, consists of 24 elements, and represents vertices of a
distorted truncated octahedron, shown in Fig.~\ref{to}. The
structure of this component will be explained in more detail in
Section~\ref{knotsSS}.

\begin{figure}[htbp!]
\begin{center}
\hfill\includegraphics*[width=0.25\textwidth]{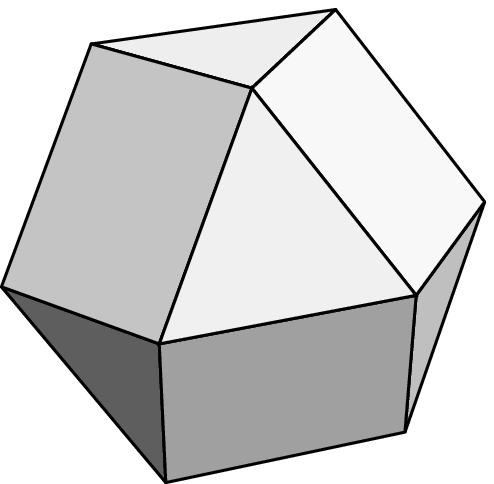}
\hfill
\includegraphics*[width=0.25\textwidth]{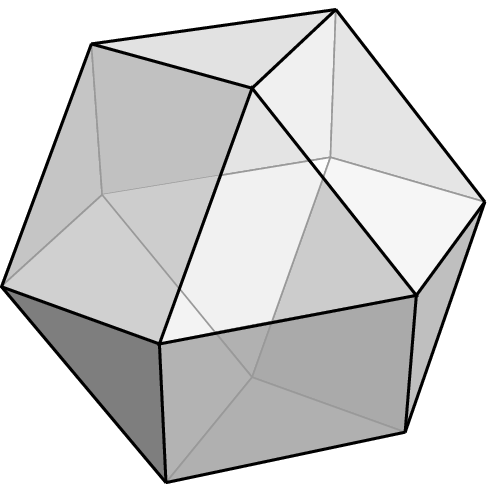}
\hfill~
\caption{Left: the form of a polyhedron with vertices ${\bD}(a)$, $a\not= 0$;
Right: the same polyhedron transparent.\label{tc}}
\end{center}
\end{figure}

\begin{figure}[htbp!]
\begin{center}
\hfill\includegraphics*[width=0.25\textwidth]{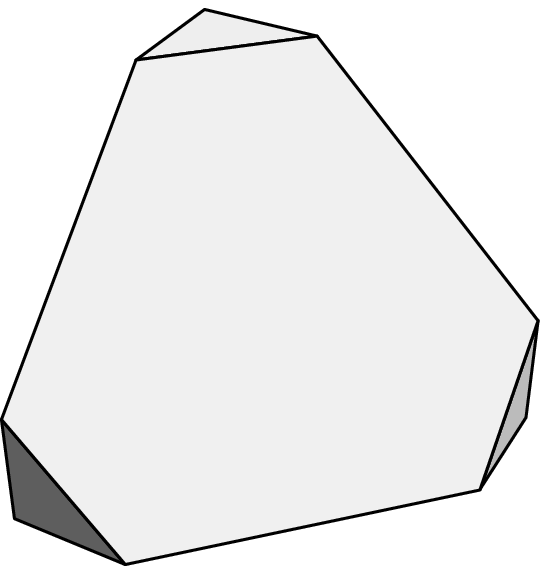}
\hfill
\includegraphics*[width=0.25\textwidth]{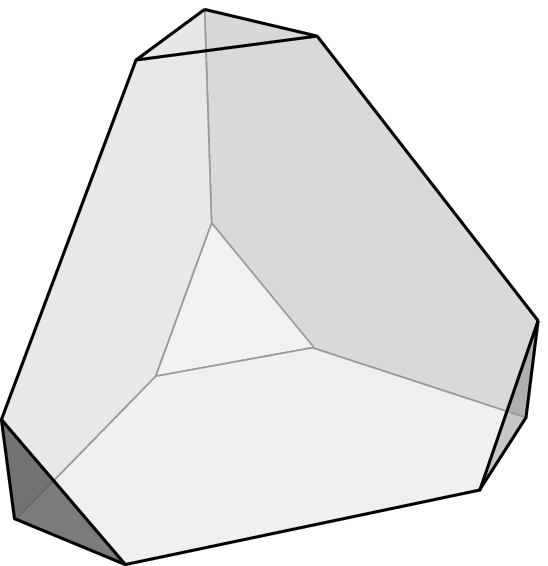}
\hfill~
\caption{Left: the form of  polyhedrons with vertices ${\bD}(a,-a)$,
${\bD}(a,2a)$, or ${\bD}(a,a/2)$, $a\not= 0$;
Right: the same polyhedron transparent.\label{tt}}
\end{center}
\end{figure}

\begin{figure}[htbp!]
\begin{center}
\hfill\includegraphics*[width=0.25\textwidth]{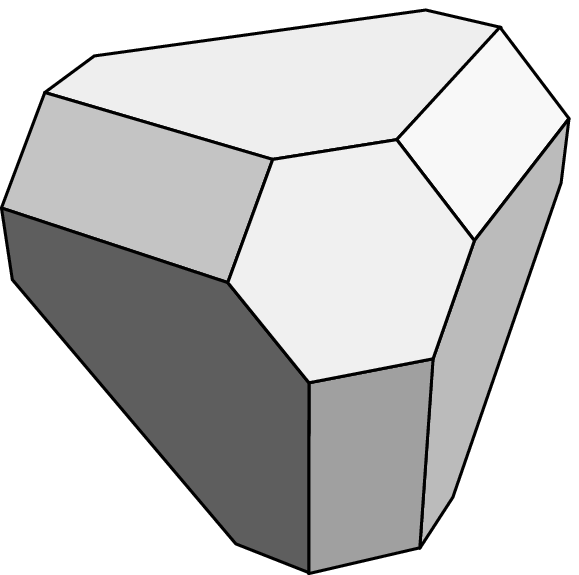}
\hfill
\includegraphics*[width=0.25\textwidth]{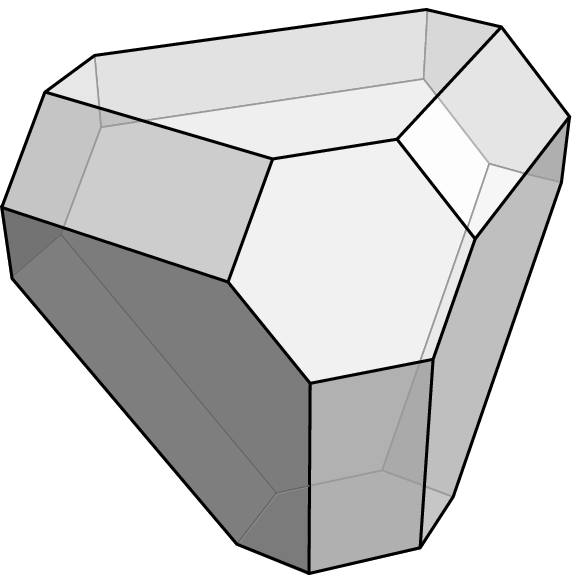}
\hfill~
\caption{Left: a typical form of a generic polyhedron
with vertices
${\bD}(a,b)$;
Right: the same polyhedron transparent.\label{to}}
\end{center}
\end{figure}

We formulate also a corollary of Proposition~\ref{trans}. For a set
${\bD}$ of discrepancies we denote by $G({\bD})$ the collection of
elements of the form $\calD G^{(i)}$, $\calD\in {\bD}$, $i=1,
\ldots , 12$.

\begin{corollary}\label{Dfini2} The equality
\[
G({\bD}(a,b)) ={\bD}(a,b))\bigcup {\bD}(a)\bigcup {\bD}(b)\bigcup {\bD}(a-b),
\]
holds for $a\not= b$.  Also $G({\bD}(a))={\bD}(a)\bigcup (0,0,0)$.
\end{corollary}

Some discrepancy triplets do not belong to any connected component;
however any element of the form $\calD G^{(i)}$ must belong to a
connected component. More precisely:
\begin{proposition}\label{colP}
The  following inclusions hold:
\begin{alignat*}{2}
(a,b,c)G^{(1,2)}&\in {\bD}(c,-a+b),\quad&(a,b,c)G^{(3,4)}&\in {\bD}(a-c,b),\\
(a,b,c)G^{(5,6)}&\in {\bD}(-b+c,a),\quad&(a,b,c)G^{(7,8)}&\in {\bD}(c,b),\\
(a,b,c)G^{(9,10)}&\in {\bD}(a,-c),\quad&(a,b,c)G^{(11,12)}&\in
{\bD}(a,b).
\end{alignat*}
\end{proposition}

\begin{proof} This assertion may be proved by inspection.
\end{proof}

\subsection{Incremental Dynamics}\label{S-incdyn}
For a given sextuple $\calR$ we denote by $\calR'$ the triplet of
the first three components of $\calR$: $
\calR'=\left(r_{\textrm{\$}\textrm{\euro}},
r_{\textrm{\$}\textrm{\pounds}},
r_{\textrm{\$}\textrm{\yen}}\right)$.  Denote further $
\calI(\calR,\Hat\calA)=\log (\calR\Hat\calA)'-\log\calR', $ where
$\Hat\calA$ is a strong arbitrage.

\begin{proposition}
\label{incrP} $\calI(\calR,\Hat\calA)$ depends only on $\Hat\calA$
and $\calD(\calR)$ and may be described as follows:
\begin{alignat*}{3}
\calI(\calR,\Hat\calA^{(1)})&=d(\calR)H^{(1)}=&-d_{\textrm{\euro}\textrm{\pounds}}(\calR)\left(
1,0,0 \right), \\
\calI(\calR,\Hat\calA^{(2)})&=d(\calR)H^{(2)}=&-d_{\textrm{\euro}\textrm{\yen}}(\calR)\left( 1,0,0 \right) , \\
\calI(\calR,\Hat\calA^{(3)})&=d(\calR)H^{(3)}=&\hphantom{-}d_{\textrm{\euro}\textrm{\pounds}}(\calR)\left(0,
1,0 \right),\\
\calI(\calR,\Hat\calA^{(4)})&=d(\calR)H^{(4)}=&-d_{\textrm{\pounds}\textrm{\yen}}(\calR)\left(0, 1,0 \right), \\
\calI(\calR,\Hat\calA^{(5)})&=d(\calR)H^{(5)}=&\hphantom{-}d_{\textrm{\euro}\textrm{\yen}}(\calR)
\left( 0,0,1\right), \\
\calI(\calR,\Hat\calA^{(6)})&=d(\calR)H^{(6)}=&\hphantom{-}d_{\textrm{\pounds}\textrm{\yen}}(\calR)
\left( 0,0,1\right).
\end{alignat*}
Also the equalities
$\calI(\calR,\Hat\calA^{(i)})=d(\calR)H^{(i)}=(0,0,0)$ hold for
$i=7,8,9,10,11,12$.
\end{proposition}

\begin{proof}
Follows from Corollary~\ref{blockC}.
\end{proof}

\subsection{Proof of Theorem~\ref{arbH}}\label{S-proofT}

This proceeds by graphing the detailed dynamics of the arbitrage
discrepancies. In this section we use the shorthand notation
$\calD_{i}$ instead of $\calD_{i}(a)$.

\begin{lemma}
\label{specialC} For any initial exchange rate ensemble belonging
to the list \eqref{list}, and for any arbitrage chain, the
corresponding sequence of discrepancies includes only elements from
the union $\calD_{0} \bigcup {\bD}(a)$, $a=\log\alpha$, see
\eqref{12disc}. The possible transition paths, arising from the
strong arbitrages listed in Table~\ref{tab1}, are given in
Table~\ref{tab3}.

Figure~\ref{GraphLeha1} plots the corresponding graph.
Figure~\ref{GraphLeha1_24} plots a similar graph, where the numbers
of the arbitrages from Table~\ref{tab1} are included, instead of
the numbers of strong arbitrages.
\end{lemma}

\begin{proof}
By inspection follows from Proposition~\ref{disc1P}.
\end{proof}

Ignoring the zero vertex $\calD_0$, the edges that lead to this
vertex and directions of the edges, another, polyhedral,
representation of the graph plotted in Fig.~\ref{GraphLeha1} is
given in Fig.~\ref{GraphLeha2}. The corresponding polyhedron is a
distorted triangular orthobicupola, shown in Fig.~\ref{tc}. The
incidence matrix $I$ of the graph plotted in Fig.~\ref{GraphLeha2}
is as follows:
\[
I=
\left(
\begin{smallmatrix-mod}%{cccccccccccc}
1& 1& 0& 0& 1& 0& 0& 0& 1& 0& 0& 1\\
1& 1& 1& 1& 0& 0& 0& 0& 0& 0& 0& 1\\
0& 1& 1& 1& 0& 0& 1& 0& 0& 0& 1& 0\\
0& 1& 1& 1& 1& 1& 0& 0& 0& 0& 0& 0\\
1& 0& 0& 1& 1& 1& 0& 0& 1& 0& 0& 0\\
0& 0& 0& 1& 1& 1& 1& 1& 0& 0& 0& 0\\
0& 0& 1& 0& 0& 1& 1& 1& 0& 0& 1& 0\\
0& 0& 0& 0& 0& 1& 1& 1& 1& 1& 0& 0\\
1& 0& 0& 0& 1& 0& 0& 1& 1& 1& 0& 0\\
0& 0& 0& 0& 0& 0& 0& 1& 1& 1& 1& 1\\
0& 0& 1& 0& 0& 0& 1& 0& 0& 1& 1& 1\\
1& 1& 0& 0& 0& 0& 0& 0& 0& 1& 1& 1
\end{smallmatrix-mod}
\right) .
\]

\begin{table}[!htbp]
\caption{Transition, caused by strong arbitrages from
Table~\ref{tab1}}\label{tab3}
\begin{tabular}
{l|llllllllllll}
 & $ \calD_{1} $ & $ \calD_{2} $ & $ \calD_{3} $ & $ \calD_{4} $ & $ \calD_{5} $ & $ \calD_{6} $ & $ \calD_{7} $ & $ \calD_{8} $ & $ \calD_{9} $ & $ \calD_{10} $ & $ \calD_{11} $ & $ \calD_{12}$ \\
\hline\\
$\Hat\calA^{(1)} $ & $ \calD_{1} $ & $ \calD_{12} $ & $ \calD_{11} $ & $ \calD_{0} $ & $ \calD_{5} $ & $ \calD_{6} $ & $ \calD_{7} $ & $ \calD_{6} $ & $ \calD_{5} $ & $ \calD_{0} $ & $ \calD_{11} $ & $ \calD_{12}$ \\
$\Hat\calA^{(2)} $ & $ \calD_{1} $ & $ \calD_{2} $ & $ \calD_{3} $ & $ \calD_{0} $ & $ \calD_{9} $ & $ \calD_{8} $ & $ \calD_{7} $ & $ \calD_{8} $ & $ \calD_{9} $ & $ \calD_{0} $ & $ \calD_{3} $ & $ \calD_{2}$ \\
$\Hat\calA^{(3)} $ & $ \calD_{1} $ & $ \calD_{0} $ & $ \calD_{7} $ & $ \calD_{6} $ & $ \calD_{5} $ & $ \calD_{6} $ & $ \calD_{7} $ & $ \calD_{0} $ & $ \calD_{1} $ & $ \calD_{12} $ & $ \calD_{11} $ & $ \calD_{12}$ \\
$\Hat\calA^{(4)} $ & $ \calD_{1} $ & $ \calD_{0} $ & $ \calD_{3} $ & $ \calD_{4} $ & $ \calD_{5} $ & $ \calD_{4} $ & $ \calD_{3} $ & $ \calD_{0} $ & $ \calD_{9} $ & $ \calD_{10} $ & $ \calD_{11} $ & $ \calD_{10}$ \\

$\Hat\calA^{(5)} $ & $ \calD_{1} $ & $ \calD_{2} $ & $ \calD_{3} $ & $ \calD_{2} $ & $ \calD_{1} $ & $ \calD_{0} $ & $ \calD_{7} $ & $ \calD_{8} $ & $ \calD_{9} $ & $ \calD_{8} $ & $ \calD_{7} $ & $ \calD_{0}$ \\
$\Hat\calA^{(6)} $ & $ \calD_{5} $ & $ \calD_{4} $ & $ \calD_{3} $ & $ \calD_{4} $ & $ \calD_{5} $ & $ \calD_{0} $ & $ \calD_{11} $ & $ \calD_{10} $ & $ \calD_{9} $ & $ \calD_{10} $ & $ \calD_{11} $ & $ \calD_{0}$ \\
$\Hat\calA^{(7)} $ & $ \calD_{1} $ & $ \calD_{1} $ & $ \calD_{0} $ & $ \calD_{5} $ & $ \calD_{5} $ & $ \calD_{6} $ & $ \calD_{7} $ & $ \calD_{7} $ & $ \calD_{0} $ & $ \calD_{11} $ & $ \calD_{11} $ & $ \calD_{12}$ \\
$\Hat\calA^{(8)} $ & $ \calD_{2} $ & $ \calD_{2} $ & $ \calD_{0} $ & $ \calD_{4} $ & $ \calD_{4} $ & $ \calD_{6} $ & $ \calD_{8} $ & $ \calD_{8} $ & $ \calD_{0} $ & $ \calD_{10} $ & $ \calD_{10} $ & $ \calD_{12}$ \\

$\Hat\calA^{(9)} $ & $ \calD_{1} $ & $ \calD_{2} $ & $ \calD_{3} $ & $ \calD_{3} $ & $ \calD_{0} $ & $ \calD_{7} $ & $ \calD_{7} $ & $ \calD_{8} $ & $ \calD_{9} $ & $ \calD_{9} $ & $ \calD_{0} $ & $ \calD_{1}$ \\
$\Hat\calA^{(10)} $ & $ \calD_{12} $ & $ \calD_{2} $ & $ \calD_{3} $ & $ \calD_{3} $ & $ \calD_{0} $ & $ \calD_{6} $ & $ \calD_{6} $ & $ \calD_{8} $ & $ \calD_{9} $ & $ \calD_{10} $ & $ \calD_{0} $ & $ \calD_{12}$ \\
$\Hat\calA^{(11)} $ & $ \calD_{0} $ & $ \calD_{3} $ & $ \calD_{3} $ & $ \calD_{4} $ & $ \calD_{5} $ & $ \calD_{5} $ & $ \calD_{0} $ & $ \calD_{9} $ & $ \calD_{9} $ & $ \calD_{10} $ & $ \calD_{11} $ & $ \calD_{10}$ \\
$\Hat\calA^{(12)} $ & $ \calD_{0} $ & $ \calD_{2} $ & $ \calD_{2} $
& $ \calD_{4} $ & $ \calD_{6} $ & $ \calD_{6} $ & $ \calD_{0} $ & $
\calD_{8} $ & $ \calD_{8} $ & $ \calD_{10} $ & $ \calD_{12} $ & $
\calD_{12}$
\end{tabular}
\end{table}

\begin{figure}[!htbp]
\begin{center}
\includegraphics*{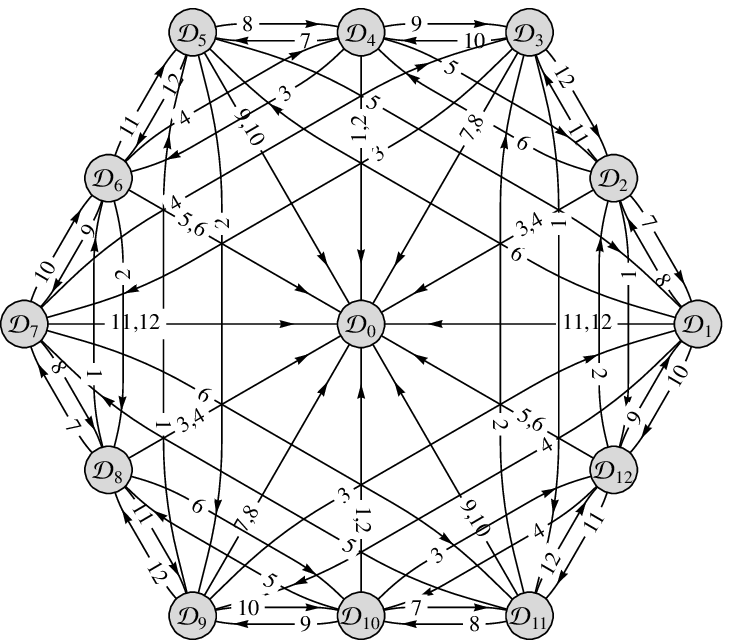}
\caption{Graph of the transitions caused by the strong arbitrages.
\label{GraphLeha1}}
\end{center}
\end{figure}

\begin{figure}[!htbp]
\begin{center}
\includegraphics*{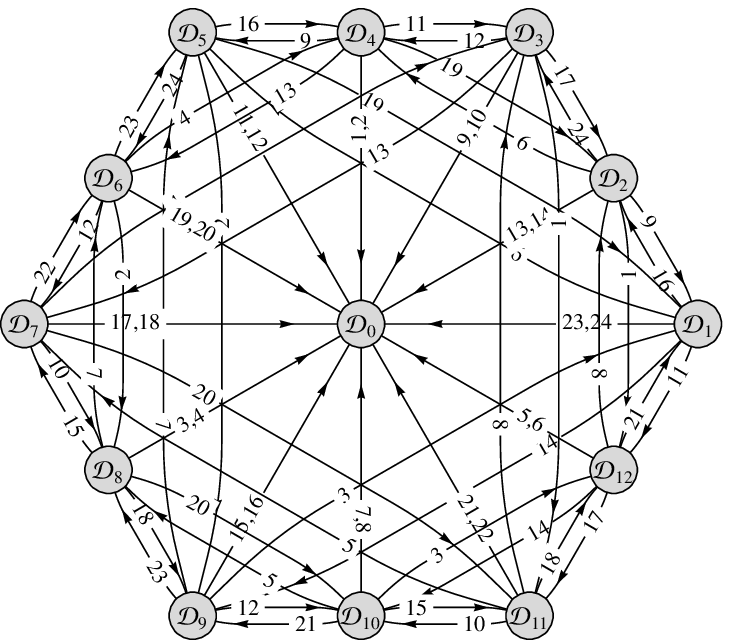}
\caption{The previous graph with the arbitrage numbers, instead of
the strong arbitrage numbers. \label{GraphLeha1_24}}
\end{center}
\end{figure}

\begin{figure}[!htbp]
\begin{center}
\includegraphics*{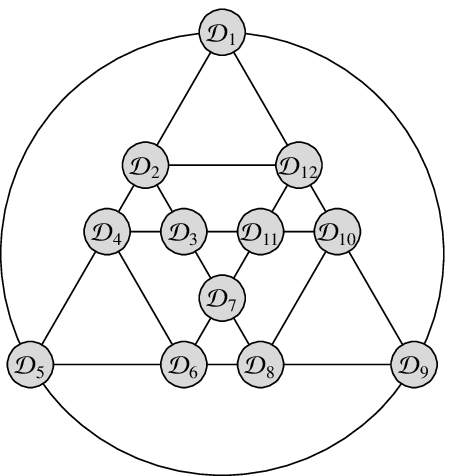}
\caption{The polyhedral representation of the principal
graph.\label{GraphLeha2}}
\end{center}
\end{figure}

Now let us deal with the coupled discrepancies and the incremental
dynamics.
\begin{corollary}\label{specialIC} For any arbitrage chain the
corresponding sequence of increments includes only the zero triplet
$\calI_{0}=(0,0,0)$ or one of the following six triplets:
\begin{alignat*}{3}
\calI_{1}&=a(1,\hphantom{-}0,\hphantom{-}0),\quad&
\calI_{2}&=a({-}1,\hphantom{-}0,\hphantom{-}0),\quad&
\calI_{3}&=a(0,\hphantom{-}1,\hphantom{-}0),\\
\calI_{4}&=a(0,{-}1,\hphantom{-}0),\quad&
\calI_{5}&=a(\hphantom{-}0,\hphantom{-}0,\hphantom{-}1),\quad&
\calI_{6}&=a(0,\hphantom{-}0,{-}1).\\
\end{alignat*}
\end{corollary}

The dynamics of the increments $\calI$ is conveniently visualised
in Fig.~\ref{GraphLeha3}.
\begin{figure}[!htbp]
\begin{center}
\includegraphics*{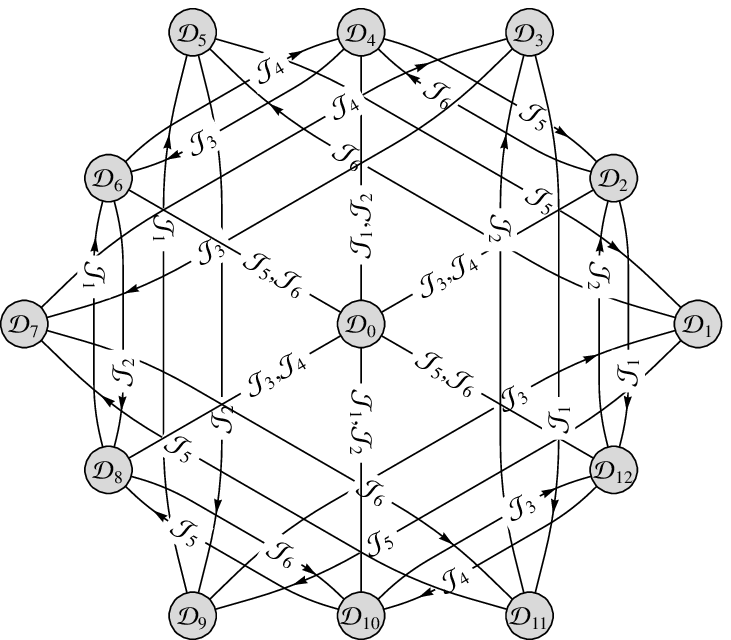}
\caption{The increment dynamics graph\label{GraphLeha3}}
\end{center}
\end{figure}

The correctness of this description of the dynamics of the
increments follows immediately from Corollary~\ref{specialC} and
Proposition~\ref{incrP}. The legitimacy of the algorithms relevant
to Theorem~\ref{arbH} and, therefore, proof of Theorem \ref{arbH}
and Proposition~\ref{arbA2H} follows from Figs.~\ref{GraphLeha1_24}
and~\ref{GraphLeha3}.

We note also that the 24-periodic chain of arbitrage from
Proposition~\ref{32} was also found looking at
Fig.~\ref{GraphLeha1_24} and~\ref{GraphLeha3}. The corresponding
route is quite natural from this perspective, and is given by
\begin{equation}\label{geomrout}
\begin{alignedat}{17}
\calD_{10}&\to&~\calD_{11}&\to&~\calD_{10}&\to&~\calD_{12}&\to&~\calD_{1}&\to& ~\calD_{12}&\to&
~\calD_{2}&\to&~\calD_{3}&\to&\\
\calD_{2}&\to&\calD_{4}& \to& \calD_{5}&\to&\calD_{4}&\to&
\calD_{6}&\to& \calD_{7}&\to& \calD_{6}& \to& \calD_{8}&\to\\
\calD_{9}&\to& \calD_{8}&\to&
\calD_{10}&\to&\calD_{8}&\to& \calD_{6}&\to& \calD_{4}&\to& \calD_{2}&\to&~\calD_{12}&\to&~\calD_{10}.
\end{alignedat}
\end{equation}

\subsection{Commuters, Terminals and Knots}\label{knotsSS}
Now we move to a proof of Theorem~\ref{irratBC} and Proposition
\ref{arbBH}. The case $d_{\textrm{\euro}\textrm{\pounds}}=
d_{\textrm{\euro}\textrm{\yen}}$ has been considered in
Section~\ref{directSS}. Thus we can assume that
$d_{\textrm{\euro}\textrm{\pounds}}\not=
d_{\textrm{\euro}\textrm{\yen}}$.

The focus is again on the dynamics of the exchange rate
discrepancies. The set  of all discrepancies that may be achievable
from $\calD=(a, b, 0)$ contains altogether 61 different elements,
see Corollary~\ref{Dfini2}. The corresponding connected component
${\bD}(a ,b)$, which contains $\calD=(a, b, 0)$, see
\eqref{24disc}, contains 24 elements listed in \eqref{24disc}. To
describe the detailed structure of this set we will introduce a new
notation. The set ${\bD}(a ,b)$,  contains six elements that have
all three components that are non-zero, and we re-denote these
elements by
\begin{alignat*}{3}
C_{1} &=\left(a,b,-a+b\right),& C_{2} &=\left(-a+b,b,a\right),&
C_{3} &=\left(a,a-b,-b\right),\\
C_{4} &=\left(-a+b,-a,-b\right),\quad& C_{5}
&=\left(-b,a-b,a\right),\quad& C_{6} &=\left(-b,-a,-a+b\right).
\end{alignat*}
We call these ensembles \emph{commuters} by way of analogy with
passenger travel.

We call an element with two non-zero components a \emph{terminal},
if $d_1\not= \pm d_2$.  There are altogether 18 terminals in
${\bD}(a,b)$. To each commuter $C_{i}$, $i=1,\ldots, 6$, we relate
three \emph{terminals} $T^{j}_{i}$, $j=1,2,3$, as follows:
\begin{alignat*}{3}
T^{1}_{1}&=\left(0,b,-a+b\right), &
T^{2}_{1}&=\left(a,0,-a+b\right), &
T^{3}_{1}&=\left(a,b,0\right);\\
T^{1}_{2}&=\left(0,b,a\right), & T^{2}_{2}&=\left(-a+b,0,a\right),
&
T^{3}_{2}&=\left(-a+b,b,0\right);\\
T^{1}_{3}&=\left(0,-a,-b\right),&
T^{2}_{3}&=\left(-a+b,0,-b\right),&
T^{3}_{3}&=\left(-a+b,-a,0\right);\\
T^{1}_{4}&=\left(0,a-b,-b\right), &
T^{2}_{4}&=\left(a,0,-b\right),&
T^{3}_{4}&=\left(a,a-b,0\right);\\
T^{1}_{5}&=\left(0,a-b,a\right), & T^{2}_{5}&=\left(-b,0,a\right),
&
T^{3}_{5}&=\left(-b,a-b,0\right);\\
T^{1}_{6}&=\left(0,-a,-a+b\right),\quad&
T^{2}_{6}&=\left(-b,0,-a+b\right),\quad&
T^{3}_{6}&=\left(-b,-a,0\right).
\end{alignat*}

\begin{lemma}
\label{commuL} The equalities
\begin{alignat*}{2}
  C_{i}G^{(7)}&=T^{1}_{i},&\quad
C_{i}H^{(7)}&=(0,0,0),\\
  C_{i}G^{(9)}&=T^{2}_{i},&\quad
 C_{i}H^{(9)}&=(0,0,0),\\
  C_{i}G^{(11)}&=T^{3}_{i},&\quad
 C_{i}H^{(11)}&=(0,0,0)
\end{alignat*}
hold for $i=1,\ldots,6$.  Also the following equalities hold: $
T^{j}_{i}G^{(k)}=C_{i}$, for $i=1,\ldots, 6$, $j=1,2,3$,
$k=8,10,12$.
\end{lemma}

We group the commuters and terminals in six knots, $K_1, \ldots,
K_6$ as follows:
\[
K_{i}=\left\{ C_{i},T^{1}_{i},T^{2}_{i},T^{3}_{i}\right\},\quad
i=1,\ldots, 6.
\]
Figure~\ref{CommuterLeha} illustrates behaviour at a knot.
\begin{figure}[!htbp]
\begin{center}
\includegraphics*{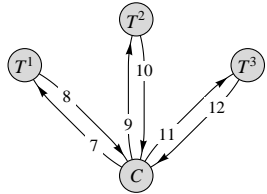}
\caption{The ``commuter--terminals'' graph of a
knot\label{CommuterLeha}}
\end{center}
\end{figure}

\subsection{Travel Between Knots}\label{S-knots}
Departing from a particular terminal, and applying some arbitrages
with numbers $k=7, \ldots, 12$, one can travel to another terminal
belonging to a different knot, simultaneously ``loading some
cargo'' upon the corresponded triplet $\calR'$. Details are given
in the following proposition.

\begin{proposition}
\label{cargoP} The following groups of equalities hold:
{\footnotesize\[ \left\{
\begin{alignedat}{4}
T^{1}_{1}G^{(3)}&=T^{1}_{2},\quad &
T^{1}_{1}H^{(3)}&=(0,a,0),\quad&
T^{1}_{1}G^{(5)}&=T^{2}_{4},\quad &  T^{1}_{1}H^{(5)}&=(0,0,b);\\
T^{2}_{1}G^{(1)}&=T^{1}_{6},\quad &
T^{2}_{1}H^{(1)}&=(-a,0,0),\quad &
T^{2}_{1}G^{(6)}&=T^{3}_{3},\quad &  T^{2}_{1}H^{(6)}&=(0,0,-a);\\
T^{3}_{1}G^{(2)}&=T^{2}_{6},\quad &
T^{3}_{1}H^{(2)}&=(-b,0,0),\quad &
T^{3}_{1}G^{(4)}&=T^{3}_{2},\quad & T^{3}_{1}H^{(4)}&=(0,a-b,0);
\end{alignedat}
\right.\\[2mm]
\]

\[
\left\{
\begin{alignedat}{4}
T^{1}_{2}G^{(2)}&=T^{2}_{5},\quad &  T^{1}_{2}H^{(2)}&=(-b,0,0),\quad &
T^{1}_{2}G^{(4)}&=T^{3}_{1},\quad &  T^{1}_{2}H^{(4)}&=(0,-a,0);\\
T^{3}_{2}G^{(1)}&=T^{2}_{2},\quad &  T^{3}_{2}H^{(1)}&=(a-b,0,0),\quad &
T^{3}_{2}G^{(6)}&=T^{3}_{3},\quad &  T^{3}_{2}H^{(6)}&=(0,0,-a);
\end{alignedat}
\right.\\[2mm]
\]

\[
\left\{
\begin{alignedat}{4}
T^{1}_{3}G^{(2)}&=T^{2}_{4},\quad &  T^{1}_{3}H^{(2)}=(a,0,0),\quad &
T^{1}_{3}G^{(4)}&=T^{3}_{6},\quad &  T^{2}_{3}H^{(4)}=(0,b,0);
\end{alignedat}
\right.\\[2mm]
\]

\[
\left\{
\begin{alignedat}{4}
T^{1}_{4}G^{(2)}&=T^{2}_{3},\quad &  T^{1}_{4}H^{(2)}&=(-a+b,0,0),\quad &
T^{1}_{4}G^{(4)}&=T^{3}_{5},\quad &  T^{1}_{4}H^{(4)}&=(0,b,0); \\
T^{3}_{4}G^{(1)}&=T^{1}_{3},\quad &  T^{3}_{4}H^{(1)}&=(-a,0,0),\quad &
T^{3}_{4}G^{(6)}&=T^{3}_{1},\quad &  T^{3}_{4}H^{(6)}&=(0,0,-b);
\end{alignedat}
\right.\\[2mm]
\]

\[
\left\{
\begin{alignedat}{4}
T^{1}_{5}G^{(2)}&=T^{2}_{2},\quad &  T^{1}_{5}H^{(2)}&=(-a+b,0,0),\quad &
T^{1}_{5}G^{(4)}&=T^{3}_{4},\quad &  T^{1}_{5}H^{(4)}&=(0,-a,0); \\
T^{3}_{5}G^{(1)}&=T^{1}_{2},\quad &  T^{3}_{5}H^{(1)}&=(b,0,0),\quad &
T^{3}_{5}G^{(6)}&=T^{3}_{6},\quad &  T^{3}_{5}H^{(6)}&=(0,0,a);
\end{alignedat}
\right.\\[2mm]
\]

\[
\left\{
\begin{alignedat}{4}
T^{1}_{6}G^{(2)}&=T^{2}_{1},\quad &  T^{1}_{6}H^{(2)}&=(a,0,0),\quad&
T^{1}_{6}G^{(4)}&=T^{3}_{3},\quad &  T^{1}_{6}H^{(4)}&=(0,a-b,0).
\end{alignedat}
\right.
\]}
\end{proposition}

We introduce the ``travel between knots'' directed graph $\Gamma$,
shown in Fig.~\ref{GraphLehaCargo}, as follows. This graph has $6$
vertices that correspond to the knots $K_1, \ldots, K_6$. A knot
$K_i$ is connected by an arrow with another knot $K_j$ if one of
terminals belonging to $K^{j}$ figures in the rows belonging to the
$i$-th subset of equalities from Proposition~\ref{cargoP}. For
instance, the knot $K_{1}$ is connected with
$K_{2},K_{3},K_{4},K_{6}$. Moreover each arrow corresponds to the
three dimensional ``cargo vector(s)'': these vectors are related in
a natural way to the increment vectors in the equalities above. For
instance,  we attach the cargo-vectors $(0,a-b,0)$ and $(0,a,0)$ to
the $K_1\to K_2$ arrow. The incidence matrix of this graph is
written below.
\[
I(\Gamma)=\left(
\begin{array}{cccccc}
 0&1&0&1&0&1\\
 1&0&1&0&1&0\\
 0&0&0&1&0&1\\
 1&0&1&0&1&0\\
 0&1&0&1&0&1\\
 0&1&0&0&0&0
\end{array}
\right)
\]

\begin{figure}[!htbp]
\begin{center}
\includegraphics*{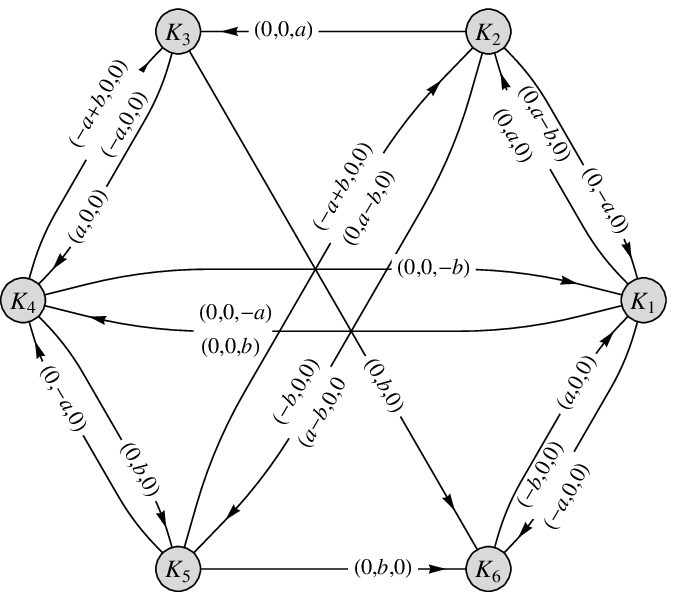}
\caption{The ``travel between knots'' graph
$\Gamma$\label{GraphLehaCargo}}
\end{center}
\end{figure}

\subsection{Finalising the proof  of Theorem~\ref{irratBC}
and Proposition~\ref{arbBH}}\label{fp2SS}
If the single transition $K_{i}\to K_{j}$ is possible we use
$W^{i\to j}$ for the corresponding cargo; we will use $W^{i\to
j}_1,W^{i\to j}_2$ if two transitions are possible. In the latter
case $W^{i\to j}_1$ refers to the upper vector indicated at graph
$\Gamma$.  For instance, $W^{1\to 2}_{1}=(0,a-b,0)$, $W^{1\to
2}_{2}=(0,a,0)$, $W^{2\to 1}=(0,-a,0)$, etc.

\begin{lemma}
\label{CyclesP} For any positive integers $N_{1},N_{2},N_{3}$ there
exists a chain ${\hbA}$ of strong arbitrages such that
$\calR{\hbA}$ has the form
\[
(r_{\textrm{\$}\textrm{\euro}}+m_{1}a-N_{1}b,
r_{\textrm{\$}\textrm{\pounds}}+m_{2}a+N_{2}b,
r_{\textrm{\$}\textrm{\yen}}+m_{3}a-N_{3}b)
\]
where $m_{1},m_{2},m_{3}$ are some positive integer numbers,
\end{lemma}

\begin{proof}
Since the moves from one terminal to another, within a particular
knot, are always possible and do not change $\calR'$ (see Lemma
\ref{commuL}), any route allowed by the graph $\Gamma$ can be
performed, and any combination of corresponding cargo can be
loaded. For the cycle $K_1\to K_{2}\to K_{5}\to K_{6}\to K_{1}$ we
have
\[
W^{1\to 2}_{1}+W^{2\to 5}_{1}+W^{5\to 6}+W^{6\to 1}=( a-b ,a,0).
\]
For the cycle $K_1\to K_{2}\to K_{3}\to K_{6}\to K_{1}$ we have
\[
W^{1\to 2}_{2}+W^{2\to 3}+W^{3\to 6}+W^{6\to 1}=( a ,a+b,a).
\]
For the cycle $K_1\to K_{2}\to K_{3}\to K_{4}\to K_{1}$ we have
\[
W^{1\to 2}_{2}+W^{2\to 3}+W^{3\to 4}+W^{4\to 1}=( a ,a,a-b).
\]
\end{proof}

\begin{corollary}
\label{abc} For any non-negative integers $N_1,N_{2},N_{3}$ and
$M_{1},M_{2},M_{3}$ there exists a chain ${\hbA}$ of strong
arbitrages such that $\calR{\hbA}$ has the form
$(r_{\textrm{\$}\textrm{\euro}}+M_{1}a-N_{1}b,
r_{\textrm{\$}\textrm{\pounds}}
+M_{2}a+N_{2}b,r_{\textrm{\$}\textrm{\yen}}+M_{3}a-N_{3}b)$.
\end{corollary}

\begin{proof}
From the lemma above it follows that we can achieve the state
\[
(r_{\textrm{\$}\textrm{\euro}}+m_{1}a-(N_{1}-1)b,
r_{\textrm{\$}\textrm{\pounds}}
+m_{2}a+(N_{2})b,r_{\textrm{\$}\textrm{\yen}}
+m_{3}a-(N_{3}+1)b,a,b,-a+b).
\]
Then moving to the  terminal $T^{3}_{1}$ and applying arbitrage
$\Hat\calA^{(6)}$ we arrive at
\[
(r_{\textrm{\$}\textrm{\euro}}+(m_{1}-1)a-(N_{1})b,
r_{\textrm{\$}\textrm{\pounds}}
+m_{2}a+(N_{2})b,r_{\textrm{\$}\textrm{\yen}}
+m_{3}a-N_{3}b,0,a,0).
\]
However, from this state we can, by Proposition~\ref{arbAH}, adjust
the numbers $m_1,m_2, m_3$ to the targets $M_1,M_2, M_3$.
\end{proof}

Theorem~\ref{irratBC} and Proposition~\ref{arbBH} follow
immediately from this last corollary.

\section{Concluding Remarks}\label{S-conclude}

The key contribution of this paper is to ask what happens to
arbitrage sequences when the number of goods or assets under
consideration is four, rather than the two, or occasionally three,
usually considered. The model is illustrated with regard to a
foreign exchange market with four currencies and traders, so there
are $C_{4}^{2}=6$ principal exchange rates. Despite abstracting
from various complications -- such as transaction costs, capital
requirements and risk -- that are often invoked to explain the
limits to arbitrage, we find that the arbitrage operations
conducted by the FX traders can generate periodicity or more
complicated behaviour in the ensemble of exchange rates, rather
than smooth convergence to a ``balanced'' ensemble where the law of
one price holds.

We use the fiction of an Arbiter, who knows all   the actual
exchange rates and what a balanced ensemble would be, to bring out
the information problem. FX traders tend to specialise in
particular currencies, so the assumption that the FX traders are
initially aware only of the exchange rates for their own
``domestic'' currencies is not entirely implausible. We show that
the order in which the Arbiter reveals information to individual
traders regarding discrepancies in exchange rate ensembles makes a
key difference to the arbitrage sequences that will be pursued. The
sequences are periodic in nature, and show no clear signs of
convergence on a balanced ensemble of exchange rates. The Arbiter
might know the law of one price exchange rate ensemble, but the
traders have little chance of stumbling onto such an ensemble by
way of their arbitrage operations.

The analysis in the present paper raises several issues to pursue
in future research. An obvious extension is to allow for a larger
number of currencies and ask what happens to the arbitrage
sequences as this number becomes large. One interesting
modification of the analysis would allow the FX traders to learn
that arbitrage sequences tend to be periodic and modify their
arbitrage strategies to take the periodicity into account. Another
modification would allow some arbitrage operations to be pursued
simultaneously, and ask what happens as the limiting case where all
arbitrage operations are exploited simultaneously is approached. An
alternative reformulation of the analysis would be as a Markov
process where the states are sextuples of exchange rates between
the four currencies and the passages between the states reflect the
effects of arbitrage operations pursued. It would be interesting to
see if this could be done without compromising the relative
simplicity of the present formulation. Finally, but by no means
exhaustively, it would be interesting to work with high frequency
data sets to test for the existence of the types of arbitrage
sequences postulated in the present paper.

 \providecommand{\nosort}[1]{} \providecommand{\bbljan}[0]{January}
  \providecommand{\bblfeb}[0]{February} \providecommand{\bblmar}[0]{March}
  \providecommand{\bblapr}[0]{April} \providecommand{\bblmay}[0]{May}
  \providecommand{\bbljun}[0]{June} \providecommand{\bbljul}[0]{July}
  \providecommand{\bblaug}[0]{August} \providecommand{\bblsep}[0]{September}
  \providecommand{\bbloct}[0]{October} \providecommand{\bblnov}[0]{November}
  \providecommand{\bbldec}[0]{December}

\end{document}